%% file: main_arxiv.tex
\documentclass[12pt]{amsart}
\usepackage{verbatim,mathtools,amssymb,amsmath}
\usepackage{subcaption}
\usepackage{amsaddr}

\usepackage{booktabs} \usepackage{multirow}
\usepackage{setspace}
\usepackage[T1]{fontenc}
\usepackage[text={6in,8in},centering]{geometry}

\usepackage{lmodern}
\usepackage[colorlinks=true, pdfstartview=FitV, linkcolor=blue,citecolor=blue, urlcolor=blue]{hyperref}

\usepackage{tikz}
\usepackage{xcolor}

\usepackage{pgfplots}
\pgfplotsset{compat=1.18} 
\usetikzlibrary{fillbetween}

\input{macros}

\def\UPS{\mathrm{UPS}}

\makeatletter \renewcommand\paragraph{\@startsection{section}{1}{\z@}{-2ex \@plus -0.5ex \@minus -.2ex}{-1em}{\bfseries}} \makeatother

\title[Multimatum]{Compromise by ``multimatum''}\thanks{ We are grateful to seminar audiences in UCLA, the Coalition Theory Network Workshop, the Parisian Behavioral Economics symposium, the Liverpool Economic Theory workshop, the France-Spain Meeting on Microeconomic Theory, the Paris-London Economic Theory Workshop and the University of Padova. We thank Carlos Al\'os-Ferrer, Simon Board, Chris Chambers, Julien Combe, Ritesh Jain, Hervé Moulin, Ludvig Sinander, Arunava Sen, Yuki Tamura, Olivier Tercieux, and William Thomson for useful remarks and comments. We are also grateful to the France-Berkeley fund for financial support. Mat\'ias N\'u\~nez  is supported by a grant of the French National Research Agency (ANR), "Investissements d'Avenir" (LabEx Ecodec/ANR-11-LABX-0047). }

\author{Federico Echenique}\address{UC Berkeley \\ \href{mailto:fede@econ.berkeley.edu}{\textup{\texttt{fede@econ.berkeley.edu}}}}

\author{Mat\'ias N\'u\~nez}
\address{CREST, CNRS, Institut Polytechnique de Paris  \\ \href{matias.nunez@polytechnique.edu}{\textup{\texttt{matias.nunez@polytechnique.edu}}}}

\usepackage[square, authoryear]{natbib}

\begin{document}

\begin{abstract}
We propose a solution and a mechanism for two-agent social choice problems with large (infinite) policy spaces. Our solution is an efficient compromise rule between the two agents, built on a common cardinalization of their preferences. Our mechanism, the \emph{multimatum} has the two players alternate in proposing sets of alternatives from which the other must choose. Our main result shows that the multimatum fully implements our compromise solution in subgame perfect Nash equilibrium.  \\
We demonstrate the power and versatility of this approach through applications to political economy, other-regarding preferences, and facility location.
  \\ \vspace{.2cm}
\noindent\textbf{JEL Codes}: D71, D72. \\ \vspace{.2cm}
\noindent\textbf{Keywords:} Compromise, Subgame-perfect implementation, Mechanism.
\end{abstract}

\maketitle

\clearpage \thispagestyle{empty}
\begin{spacing}{1.05}
\tableofcontents    
\end{spacing}
\clearpage

\setcounter{page}{1}
\section{Introduction}
\subsection*{Motivation: \textsl{Hoc iudicium iustum est.}}

It is Rome, 59 BC. Gaius Julius Caesar and Marcus Calpurnius Bibulus are co-consuls in bitter disagreement over public policy. Historically, the dispute ended badly. Our paper proposes a way out for two-party conflicts like theirs. Our solution has two components: a social welfare criterion that defines a fair compromise and an institutional mechanism ensuring that the compromise is reached.

Any theory of compromise requires interpersonal utility comparisons. Suppose policies are represented as points $x \in [0,1]$. If Caesar prefers policy $x$ over $x'$, while Bibulus has the opposite preference (as was often the case), how do we weigh Caesar's gain against Bibulus's loss? Ordinal preferences alone cannot solve this. This fact has been shown formally by \cite{shapley1969utility}, and emphasized by \cite{myerson1977two}. Economists often assume quasilinearity over money to establish a one-for-one tradeoff between the agents' utilities. But when money is absent, utilities are conceptually and philosophically difficult to compare.  

If we map outcomes into utility values (utils), we need to make sense of choices among points in the utility possibility set. For a specific example, consider the utility possibility set depicted in Figure~\ref{fig:intro}. Imagine that we want to evaluate outcomes where Caesar obtains a utility of $0.8$ or more. An increase in Caesar's utility from $0.8$ to $1$, for instance, corresponds to a decrease for Bibulus from (about) $0.4$ to $0$: the size of $\Delta u_C$ is different from the size of $\Delta u_B$. Despite adopting the same range of values, the two utilities are measured on different scales. We cannot simply trade them off one-for-one; they are apples and oranges.

To make an apples-to-apples comparison, we need a common yardstick. We propose a specific \df{cardinalization} of the consuls' preferences. For any policy $x$, consider the lower contour set $L_i(x)$, the set of policies that consul $i \in \{C,B\}$ deems worse than or equal to $x$. We represent preferences using the uniform probability measure of these sets: $\la(L_i(x))$. The intuition is straightforward: we measure a policy's worth to an agent by asking what fraction of all possible policies are worse than it. The resulting utility possibility set is depicted on the right in Figure~\ref{fig:intro}. Note the $-1$ slope of the Pareto frontier, ensuring that the size of the sets of utilities being evaluated in any tradeoff is the same. The slope $=-1$ frontier looks like what we would get in a model with money and quasi-linear preferences.

We call $\la(L_i(x))$ a \df{canonical} utility representation. Because it maps both agents' preferences into the same shared probability space, these utilities become interpersonally comparable. We formalize the ideas in Theorem~\ref{thm:commonnu}, showing that a pair of utility representations is comparable in our sense if and only if they can be written as $\nu(L_i(x))$ for some probability measure $\nu$. Figure~\ref{fig:intro} on the right depicts the ``straightened'' utility possibility set, which corresponds to our common cardinalization (Appendix~\ref{sec:appexamples} has additional examples of this transformation).

\begin{figure}[t]
    \centering
    \begin{tikzpicture}[scale=0.5]
        \begin{axis}[
            width=0.6\textwidth,
            height=0.6\textwidth,
            title={\textbf{Utility Possibility Set}},
            xlabel={$u_C(x)$},
            ylabel={$u_B(x)$},
            xmin=0, xmax=1.1,
            ymin=0, ymax=1.1,
            axis lines=left,
            enlargelimits=false
        ]
        
                        \addplot [
            domain=0:1, 
            samples=200, 
            color=blue,
            very thick,
        ]
        {sqrt(1 - x)};

        \end{axis}
    \end{tikzpicture}
    \begin{tikzpicture}[scale=0.5]
        \begin{axis}[
            width=0.6\textwidth,
            height=0.6\textwidth,
            title={\textbf{Utility Possibility Set}},
            xlabel={$u_C(x)$},
            ylabel={$u_B(x)$},
            xmin=0, xmax=1.1,
            ymin=0, ymax=1.1,
            axis lines=left,
            enlargelimits=false
        ]
        
                \draw[thick,blue] (0,1) -- (1,0);
\draw[dotted,thin] (0,0) -- (.8,.8);
\draw[thin,purple] (0.5,.9) -- (.5,.5) -- (.9,.5);
\addplot[mark=*, black] coordinates {(1/2,1/2)};
        \end{axis}
    \end{tikzpicture}
    \caption{The utility possibility set when  $u_C(x)=1-x$ and $u_B(x)=\sqrt{x}$ (left), and under common canonical utilities (right).}
\label{fig:intro}
  \end{figure}

Armed with a common cardinalization, we can now define a fair social welfare criterion. A \df{compromise} between Caesar and Bibulus is a policy that maximizes the minimum canonical utility:
\[ \max_{x\in X} \min \{\la(L_C(x)),\la(L_B(x)) \}.\]
The compromise is shown in utility space in Figure~\ref{fig:intro}.\footnote{While we focus on the Lebesgue measure for utils on the interval $[0,1]$, we shall see that the measure on outcomes, or policies, $X$ can be quite arbitrary. If the measure is $\nu$ then the compromise is a solution to $\max_{x\in X} \min \{\nu(L_C(x)),\nu(L_B(x)) \}$.}

Defining a compromise is only half the problem; we must also ensure the consuls actually reach it. To this end, we introduce the \df{multimatum}: a mechanism, or institution, that an uninformed third party can implement to ensure that a compromise is reached. We imagine a framer setting up a constitutional mechanism, or an arbitrator choosing a bargaining protocol. The multimatum is based on a proposal followed by a single round of counterproposals. 

First, Caesar will offer a closed set $A_C$ of policies. Bibulus can choose $x\in A_C$ to be implemented or counter-propose with a closed set $A_B$. The set $A_B$ is final; take-it-or-leave-it. Caesar must choose a policy from $A_B$. This would, however,  give Bibulus all the bargaining power. It would not be a compromise at all.  To restore the balance of power, we impose one crucial volumetric rule: the measure of Bibulus's counterproposal cannot be less than the measure of Caesar's initial offer ($\la(A_B) \geq \la(A_C)$). This constraint incentivizes Caesar to offer a relatively large, generous initial $A_C$ to preempt Bibulus from replying with a small, highly targeted $A_B$. 

Our \emph{main result} (Theorem~\ref{th:implementation}) establishes that the institutional balance inherent in the multimatum works: the subgame-perfect Nash equilibrium outcomes of the multimatum coincide exactly with the set of compromise solutions.

\subsection*{Our contribution.} Our paper examines a broad class of problems, of which the example of the Roman co-consuls is a particular instance. These problems feature two agents bargaining over a common policy or outcome; they are widespread in the political economy and public goods literature. A (very incomplete) list of such models is: \cite{rubinstein1982perfect}, \cite{AlesinaRosenthal1995}, \cite{AlesinaTabellini1990}, \cite{McCartyPooleRosenthal2006}, \cite{woonanderson2012}, \cite{bowen2014mandatory}, \cite{EraslanEvdokimovZapal2020}, and \cite{alikartikkleiner23}.

There are two sides to our paper's \textit{sestertius}. One side is the rule we have described. It reflects a compromise between two parties, Caesar and Bibulus, in the Roman example. Crucially, the rule features a common cardinalization of the two parties' ordinal utilities, which makes cardinal comparisons more palatable. Outside of highly structured environments with quasi-linear preferences over money or lotteries and expected utility, the literature is silent on how to tackle the comparability of utilities. We characterize (Theorem~\ref{thm:commonnu}) the utility representations that allow for interpersonal comparisons of utility, in a sense that generalizes the slope $=-1$ property from our previous discussion. Our proposal offers a pragmatic way of comparing utilities; and the compromise rule itself is, we believe, normatively attractive and intuitively captures the notion of a fair middle ground. 

The other side of our proposal deals with the implementation of the rule. How would an impartial framer or arbitrator (perhaps the senate in the Roman example) ensure that the rule is followed and the negotiation not devolve into barbarism? Our answer is the simple multimatum mechanism. Importantly, the framer does not need to know or elicit the consuls' preferences. The senate can simply enforce the protocol while remaining ignorant of the final objectives of the two consuls. The end result is guaranteed to be a compromise solution. The senate is also reassured that any equilibrium outcome corresponds to the desired rule; there is no need to impose an equilibrium selection criterion, as is so often the case in mechanism design.

We should emphasize the dual role played by the measure $\nu$. The measure is a crucial component of the definition of the rule, as it provides a common cardinal scale on which to make interpersonal utility comparisons. It is also the key idea behind the countervailing forces built into the negotiation protocol. The presence of $\nu$ on both sides of our proposal is what makes it work.

In the paper, we work out two detailed applications. The first is to the political economy model of \cite{bowen2014mandatory}, which features two political parties bargaining over the level of a public good and individual private consumption. We calculate the compromise solution for this model and carry out comparative statics when we vary the parties preferences. The second application is to a model of pure private consumption with other-regarding preferences as in \cite{fehr1999theory}. These preferences effectively render the model one of public goods, as each agent cares about the final allocation to both parties. Again, we can perform comparative statics to understand how the compromise solution depends on the agents' preferences.

In addition to the two detailed applications, we illustrate our results with simple examples of single-peaked preferences, facility-location problems, and random social choice with expected-utility preferences. For each of these widely-studied environments, our proposal provides an appealing solution.

Finally, we suggest a positive interpretation of the multimatum. Our bargaining protocol seems quite natural, and as ``realistic'' as other well known non-cooperative models of bargaining. In practice, real-world bargaining typically unfolds over a limited number of stages, making it reasonable to assume some fixed rule governing when negotiations end. Such rules are common, for example, in the interactions between bicameral legislative bodies. Taken together, our model may therefore provide a plausible representation of how at least some political bargaining actually occurs. 

The rest of the paper is structured as follows. Section~\ref{sec:preferences} introduces the preferences that we consider, derives the basic properties used throughout, and lays down the model. Section~\ref{sec:compromise} characterizes the canonical utility representations as being those that allow for interpersonal comparisons, in the sense that we define in the paper. Section~\ref{sec:result} presents the mechanism, discusses the main result dealing with implementation, and explores several applications. Detailed applications are in Section~\ref{sec:applications}. Section~\ref{sec:discussion} relates our findings to the insights from the implementation literature and the bargaining literature. Section~\ref{sec:conclusion} concludes.

\section{Model\label{sec:preferences}}

\subsection{Preferences in metric spaces}

Let $(X,\rho)$ be a compact and connected metric space, and $\Sigma$ its Borel $\sa$-algebra.  Given is a probability measure $\nu$ on $(X,\Sigma)$. Suppose that $\nu$ has \df{full support} (meaning there is no proper closed subset of $X$ that has probability one) and is non-atomic. The probability space $(X,\Sigma,\nu)$ is fixed throughout the paper. 

A binary relation $\succeq$ on $X$ is a \df{preference relation} if it is complete and transitive.

Consider the following definitions.
The \df{lower contour set} of a preference $\succeq$ at $x\in X$ is the set $L_{\succeq}(x)=\{x'\in X:x\succeq x'\}$; its
   \df{upper contour set} at $x\in X$ is $U_{\succeq}(x)=\{x'\in X:x'\succeq x\}$; and its 
   \df{indifference set} at $x\in X$ is $I_{\succeq}(x)=L_{\succeq}(x)\cap U_{\succeq}(x)$.
  A preference $\succeq$ is \df{continuous} if, for any $x\in X$, its lower and upper contour sets at $x$ are closed. A preference relation $\succeq$ on $X$ is \df{locally strict} if whenever $x\succeq y$, $N_x$ is a neighborhood of $x$ and $N_y$ is a neighborhood of $y$, there exist $x'\in N_x$ and $y'\in N_y$ with $x'\succ y'$. A preference relation $\succeq$ is \df{locally non satiated} if, for any $x\in X$ and any neighborhood $N_x$ of $x$ in $X$ there exists $x'\in N_x$ with $x'\succ x$.

A continuous preference relation has \df{thin indifference curves} if its indifference curves $I_\succeq(x)$ have measure zero and $I^c_\succeq(x)$ is dense in $X$, for all $x\in X$. The denseness of $I^c_\succeq(x)$ is equivalent to $I_\succeq(x)$ having an empty interior.

These definitions are, for the most part, standard. The property of being locally strict was introduced by \cite{border1994dynamic}. The property of having thin indifference curves should be intuitive: it requires that indifferences are ``small,'' both topologically and according to the measure $\nu$. You could say that thin indifferences are both small and rare.

\subsection{Canonical utility representation}

A function $u : X \to \Re$ is a \df{utility representation} of $\succeq$ if for all $x, y \in X$, 
\[ x \succeq y \iff u(x) \ge u(y). \]

A utility representation $u$ is a \df{canonical representation} of $\succeq$ if there exists a (Borel) probability measure $\mu$ on $X$ with $u(x) = \mu(L_{\succeq}(x))$. 

By Debreu's theorem (\cite{debreu1964continuity}), our assumptions guarantee that every continuous preference $\succeq$ admits a continuous utility representation. Normalizing, we can, without loss of generality, take this utility to be surjective onto $[0,1]$. Under the same assumptions, any such utility representation can be cast in the canonical form (see Theorem~\ref{thm:commonnu}, Part 1). Furthermore, Debreu’s original construction is, in essence, constructed by measuring lower contour sets; precisely the perspective underlying our canonical form. This is why we regard the term canonical as particularly fitting.

\subsection{Convex preferences}\label{sec:convex}

When $X$ is a convex subset of $\Re^d$, we may use the added structure to introduce and exploit additional properties of a preference relation. Consequently, as we shall see, many economic applications are special cases of our model. When $X\subseteq \Re^d$ is convex, we say that $\succeq$ is \df{convex} if its upper contour sets $U_\succeq(x)$ are convex sets for all $x\in X$. We say that $\succeq$ is \df{explicitly convex} if it is convex and, moreover, for any $x,y\in X$ and $\theta \in (0,1)$, if $y\succ x$, then $\theta y + (1-\theta)x \succ x$. Finally, a preference $\succeq$ is \df{linear} if there exists $v\in\Re^d$ with the property that $x\succeq y$ iff $v\cdot x\geq v\cdot y$. 

The property of being explicitly convex is a minor strengthening of convexity. It is compatible with indifference curves that have ``flat'' segments. For example, linear preferences are explicitly convex.\footnote{The assumption of strict convexity is common in economics, but it is significantly stronger than explicit convexity and is not satisfied by linear preferences. Explicit convexity is known as the ordinal property that ensures that local optima are global optima, for which convexity alone does not suffice (see, for example, Theorem 192 in \cite{BorderMax2015}).}

Our first result establishes that many economic environments satisfy the properties that we shall impose. Indeed, many economic applications involve a convex set of outcomes and convex preferences. Proposition~\ref{prop:thinindiff} provides rather weak conditions under which our results are applicable to such environments.
  
\begin{proposition}\label{prop:thinindiff}Suppose that $X$ is a full-dimensional, convex, and compact subset of $\Re^d$ endowed with the usual Euclidean norm, and that the probability measure $\nu$ is absolutely continuous with respect to the Lebesgue measure on $\Re^d$. Let $\succeq$ be an explicitly convex and locally strict preference on $X$. Then $\succeq$ has thin indifference curves.  
\end{proposition}

The proof of Proposition~\ref{prop:thinindiff} is in Section~\ref{sec:pfpropTIC}. All the proofs in the paper are relegated to Section~\ref{sec:proofs}.
 
\subsection{The model}\label{sec:model}

 A \df{two-agent social choice problem} is a tuple ${((X,\rho,\nu),\succeq_1,\succeq_2)}$ in which $(X,\rho)$ is a compact metric space, $\nu$ is a full-support, non-atomic Borel probability measure on $(X,\rho)$; and for $i=1,2$, $\succeq_i$  is a continuous preference relation on $X$ with thin indifference curves. In our model, $X$ represents the set of possible outcomes or policies that may be chosen collectively. 
 
We often refer simply to $X$, and take $\rho$ and $\nu$ to be understood from context. For example, when $X$ is Euclidean we assume the usual norm and Lebesgue measure.
 
Let $\succeq_1$ and $\succeq_2$ be two locally strict preference relations. For each $\succeq_i$ with $i=1,2$ and each $x\in X$, we write $L_i(x)$ rather than $L_{\succeq_i}(x)$ to denote the lower contour set at $x$.

\subsection{The compromise solution}\label{sec:compromise}

Our notion of compromise relies on a canonical utility representation for each agent. Crucially, the canonical representation uses the same measure for both agents: it is a \emph{common cardinalization} of the agents' preferences.

\begin{definition}An outcome $x^*\in X$ is a \df{compromise solution} if it solves the problem:
\[
\max_{x\in X} \min\{\nu (L_1(x)), \nu (L_2(x)) \}.
\]
\end{definition}

\begin{remark} By Lemma~\ref{lem:cont} and the compactness of $X$, there exists at least one compromise solution.
\end{remark}

Note that we can always choose, as our utilities, the canonical representations defined by $\nu$. The question is what the property of a common cardinalization means. When is a pair of utility representations $(u_1,u_2)$ a pair of canonical representations with a common measure?  We argue that this property underlies a basic notion of interpersonal comparability, a generalization of the one-for-one Pareto frontier in Figure~\ref{fig:intro} that we described in the introduction. 

Consider the \df{utility possibility set}
\[
\UPS = \{(u_1(x),u_2(x)) \colon x\in X \},
\] given a pair of utility representations $u_1$ and $u_2$ for, respectively, $\succeq_1$ and $\succeq_2$.

\begin{theorem}\label{thm:commonnu}
Let $((X,\rho),\succeq_1,\succeq_2)$  be a two-agent social choice problem. Let $u_i$ be a utility representation of $\succeq_i$, $i=1,2$.
\begin{enumerate}
    \item For each $i \in \{1, 2\}$, $\succeq_i$ admits a canonical utility representation by means of a probability measure $\nu_i$. 
    \item Under the measure $\nu_i$, the push-forward measure $(u_i)_*\nu_i$ on $[0,1]$ is the Lebesgue measure.
    \item There exists a common Borel probability measure $\nu$ on $X$ (i.e., we may take $\nu_1 = \nu_2 = \nu$) satisfying properties (1) and (2) simultaneously if and only if the following condition holds: 
    for every Borel set $E \subseteq [0,1]$,
    \[ \la(E) \le \la(\UPS(E)), \]
    where $\la$ is the Lebesgue measure on $[0,1]$ and $\UPS(E)$ is given by
    \[ \UPS(E) = \{ u_2(z) \in [0,1] : z \in X \text{ and } u_1(z) \in E \}. \]
\end{enumerate}
\end{theorem}

The condition in Statement 3 of Theorem~\ref{thm:commonnu} is a direct generalization of the simple comparison we used for Figure~\ref{fig:intro} in the introduction. There, we focused on the utilities in $[0.8,1]$ for one agent and compared them with the corresponding utilities for the other agent, observing that they formed a set of measure $0.4$. This example exposed the absence of a common utility scale: the sense that the density of underlying utility values was different for each of the two agents. Theorem~\ref{thm:commonnu} formalizes this idea via a structural condition on the projections of $\UPS$ onto each player's utility coordinates. See Section~\ref{sec:properties} for a deeper discussion of this property.

\subsection{Regularity}

Our main result applies to situations where the compromise solution is, in some general sense, ``interior.''

\begin{definition}
  A problem ${((X,\rho,\nu),\succeq_1,\succeq_2)}$ is \df{regular} if, for any compromise solution $x^*$, both preferences satisfy local non-satiation at $x^*$.
\end{definition}

Regularity is discussed further in Section~\ref{sec:properties}.

\section{Main Result}\label{sec:result}

We proceed with our main result by first describing the multimatum mechanism and showing that it achieves full implementation of the compromise outcomes in subgame-perfect Nash equilibrium.

The multimatum mechanism is very simple. It asks Agent 1 to make an initial proposal that consists of a set $A_1$ of possible outcomes. Agent 2 receives the proposal and may either choose an outcome from the proposed set, or reject the set chosen by Agent 1. In the latter case, Agent 2 chooses a set of outcomes $A_2$ that must be at least as large as Agent 1's proposal (as measured by $\nu$). The game ends with Agent 1 choosing an outcome from $A_2$, the set proposed by Agent 2. The constraint that $A_2$ must have a measure at least as large as $A_1$ can be interpreted as follows: if $\nu(A_1)$ is a measure of the generosity of Agent 1's proposal, Agent 2 must be at least as generous as Agent 1 should they choose to decline $A_1$.

We lay out a formal description of the mechanism.

\begin{definition}
The \df{Multimatum Mechanism} is given by the following game form:
\begin{enumerate}
    \item Agent 1 chooses a non-empty closed set $A_1 \subseteq X$.
    \item Agent 2 observes $A_1$ and either chooses an alternative $x \in A_1$ (ending the game with outcome $x$) or rejects $A_1$.
    \item If Agent 2 rejects $A_1$, Agent 2 must propose a non-empty closed set $A_2 \subseteq X$ satisfying $\nu(A_2) \geq \nu(A_1)$.
    \item Agent 1 observes $A_2$ and chooses an alternative $x \in A_2$ (ending the game with outcome $x$).
\end{enumerate}
When coupled with two preference relations, the multimatum mechanism defines an extensive-form game.
\end{definition}

\subsection{Subgame-perfect implementation of a compromise}

We write $\mathcal{F}(X)=\{F\subseteq X\mid F\neq \emptyset, F \text{ closed in } X \}$ to denote the set of non-empty closed subsets of the outcome set $X$. Importantly, $\mathcal{F}(X)$ is compact under the Hausdorff metric topology.

A strategy profile in the game induced by the Multimatum mechanism is a pair $\sigma=(\sigma_1,\sigma_2)$.
The strategy of the first mover is $\sigma_1=(\sigma^1_1,\sigma^2_1)$, where:
\begin{itemize}
    \item $\sigma^1_1 \in \mathcal{F}(X)$ is the initial proposal.
    \item $\sigma^2_1$ is a function that maps any history $(A_1, A_2)$ (where Player 2 rejected and proposed $A_2$) to a choice $x \in A_2$.
\end{itemize}
The strategy of the second mover is a function $\sigma_2$ that maps any proposal $A_1$ to an action, which is either an acceptance $x \in A_1$ or a counter-proposal $A_2 \in \mathcal{F}(X)$ satisfying $\nu(A_2) \ge \nu(A_1)$.

The outcome of a strategy profile $\sigma$ is denoted by $\MU(\sigma)$. The game induced by the mechanism is not finite, and standard equilibrium existence results do not apply.

We say that the multimatum mechanism \df{implements the compromise options in subgame-perfect equilibrium} if:
\begin{enumerate}
    \item For any subgame-perfect Nash equilibrium $\sigma^*$, the outcome $\MU(\sigma^*)$ is a compromise alternative.
    \item Conversely, for any compromise outcome $x \in X$, there is a subgame-perfect Nash equilibrium $\sigma^*$ such that $\text{MU}(\sigma^*) = x$.
\end{enumerate}

\begin{theorem}\label{th:implementation}
In a regular problem $(X,\succeq_1,\succeq_2)$, the multimatum mechanism implements the compromise solutions in subgame-perfect Nash equilibrium.
\end{theorem}

Theorem~\ref{th:implementation} is the main result of the paper. The proof is in Section~\ref{sec:proofs}. It proceeds by first establishing that player 1 has a strategy that guarantees that, in all subsequent subgames, both players have well-defined best responses. Using this strategy for player 1, we then construct an equilibrium for any given compromise alternative. The compactness of the space of all menus (the space of non-empty closed sets) that the players may choose plays a key role in this step. Finally, we demonstrate that the outcome of every equilibrium must coincide with a compromise solution.

\subsection{Properties of the compromise}\label{sec:properties}

We proceed to discuss the ideas behind the compromise solution in some detail. First, we explain the common cardinalization further. Then, we consider properties of the compromise and look at the regularity assumption.

\paragraph*{On the common cardinalization.}

We have not solved the ethical dilemmas involved in making interpersonal comparisons. What we have is a simple and pragmatic proposal generalizing a procedure that is already widely used in practice. In market design, when $X$ is finite (a situation that is outside of our model), it is common to compare evaluate outcomes according to how they are ranked by different agents. Effectively, to compare one agent's first choice to another agent's third. This amounts to a canonical utility defined by counting the size of the lower contour sets.\footnote{In market design, match quality is often assessed by ranks as a summary welfare statistic. For example, mechanisms are evaluated by reporting how many agents receive their top-1, top-2, top-3, etc.\ options. Mechanisms such as Deferred Acceptance use only ordinal information, but their outcomes are evaluated via a specific cardinalization of these ordinal preferences. Applications include school choice \cite{abdulkadirouglu2005new}, teacher assignment \cite{combe2022design}, and course assignment at HBS \cite{budish2012multi}. On the theoretical side, see \cite{featherstone2020rank} for rank-efficiency and \cite{nikzad2022rank} for rank-optimal assignments in uniform markets.
} Our proposal extends this idea to models with infinitely many alternatives.

Consider the notion that Agent 1 obtains a utility of $v_1=0.9$ while Agent 2 obtains $v_2=0.8$. In a finite model, there might be $100$ possible outcomes, and Agent 1 is getting their top $10$ most preferred outcome while Agent 2 is getting their $20$th most preferred outcome. In an infinite model, we want $0.9$ and $0.8$ to retain their usual meaning as real numbers in the interval $[0,1]$: the mass of numbers below $0.9$ is $90\%$ of $[0,1]$, and the fraction below $0.8$ is $80\%$. This requires using the same (Lebesgue) measure to evaluate Agent 1's and Agent 2's utils.

By Parts 1 and 2 of Theorem~\ref{thm:commonnu}, such a measure is available for each agent, no matter which (continuous) utility representation we are using. But \emph{unless we are using a common cardinalization}, the meaning will be essentially different for each of the agents. They will not be compatible and therefore not comparable: for example, when we say that $v_1$ is in the top $10\%$ of utility values for Agent 1, the underlying outcomes in $X$ should also correspond to $10\%$ of the possible utility values for Agent 2. So when we compare welfare in terms of utils, utility values as real numbers on $[0,1]$, both agents are mapping these in the same terms to the underlying physical outcomes.

Under a common cardinalization, there exists a transformation from Agent 1's utility values to Agent 2's that satisfies a generalization of the slope $=-1$ property in the introduction.  In the introduction, the $\UPS$ equals the Pareto frontier, so we may take the transformation to be $T:[0,1]\to [0,1]$. In general, though, to each utility value of Agent 1, there may correspond multiple values for Agent 2. In particular, we obtain an essentially unique random correspondence of utils $v_1$ to utils $v_2$, a Markov kernel $T:[0,1]\to \Delta([0,1])$. The slope $=-1$ property corresponds to a preservation of measure, so that for any set $E$,
\[
  \la(E) = \int_{[0,1]} \int_E T(v_1, \diff v_2) \diff \la(v_1).
\]
In words, the measure of the set of $v_1$ that corresponds (under $T$) to assigning Agent 2 utils $v_2\in E$ equals the measure of $E$.

Finally, we should discuss two points. The first is trivial but worth emphasizing. We have assumed that both agents' utils are in $[0,1]$. This is an irrelevant normalization. We could, for example, have 2's utility take values in any interval and would just need to keep track of the relative sizes of the two ranges of utilities. Slope $=-1$ would no longer be the right criterion, but there is an obvious way to capture the same idea. 

The second point highlights the degrees of freedom in choosing a measure that defines a common cardinalization. By choosing such a measure, we can span all the Pareto optimal outcomes in which neither agent achieves their globally best outcome (neither agent is satiated). 

\begin{proposition}\label{prop:POmu} Let $((X,\rho,\nu),\succeq_1,\succeq_2)$ be a two-person social choice problem. 
    Let $x \in X$ be Pareto optimal such that neither agent is satiated at $x$. There exists a full-support Borel probability measure $\mu$ on $X$ that is absolutely continuous with respect to $\nu$, such that $x$ solves the compromise solution problem:
\[ x \in \arg\max_{x' \in X} \min \{ \mu(L_1(x')), \mu(L_2(x')) \}. \]
\end{proposition}

Proposition~\ref{prop:POmu} follows simply from equalizing the measure of the lower contour sets and ensuring that this can be done in a way that makes the measure absolutely continuous with respect to $\nu$. The hypothesis that neither agent is satiated is unavoidable. If Agent 1 is satiated and Agent 2 achieves a utility $\mu(L_2(x))<1$, then there exists a local improvement on the egalitarian objective in the definition of the compromise.\footnote{Let $M = \mu(L_2(x))$. By local non-satiation, there exists $x^n\to x$ with $\mu(L_2(x^n))>M$. For $x$ to be a compromise, we would need that $\mu(L_1(x^n))\leq M<1$. This is not possible for all $n$ as $L_1(x)=X$.} (Of course, if both agents are satiated at $x$, then it is a compromise for any choice of $\mu$.)

\paragraph*{Compromise and regularity.}

\begin{lemma}\label{lem:PO} If the problem is regular and $x^*$ is a compromise solution, then $x^*$ is Pareto optimal, and 
\[\nu(L_1(x^*)) = \nu (L_2(x^*)) \geq \frac{1}{2}.\] Moreover, each agent is indifferent over all compromise solutions.
  \end{lemma}
  \begin{proof}
Suppose, towards a contradiction, that $\nu(L_1(x^*)) \neq \nu (L_2(x^*))$. We may assume, without loss of generality, that $\nu(L_1(x^*)) < \nu (L_2(x^*))$. Since the problem is regular, $\succeq_1$ is locally non-satiated at $x^*$, and $x\mapsto \nu(L_i(x))$ is continuous, there exists $x'$ (in a neighborhood of $x^*$) with $\nu(L_1(x^*)) < \nu (L_1(x'))$ and $\nu(L_1(x')) < \nu (L_2(x'))$, which is absurd.

The same argument establishes that $x^*$ is Pareto optimal. For suppose (without loss of generality) that there exists $x\in X$ with $x\succeq_1 x^*$ and $x\succ_2 x^*$. That would mean that $x$ is also a compromise solution. Since the problem is regular, $\succeq_1$ must be locally non-satiated at $x$. Since $X$ is connected and $\succeq_2$ continuous, there is a neighborhood $N_x$ of $x$ such that $N_x\cap X\neq \{x\}$ and for which $x'\succ_2 x^*$ holds for all $x'\in N_x$. By local non-satiation, there is $x'\in N_x$ with $x'\succ_1 x^*$, which would contradict the definition of compromise.

Suppose, again towards a contradiction, that $\nu(L_1(x^*)) = \nu (L_2(x^*)) <  \frac{1}{2}$. Since $\nu(X)=1$, this means that: 
\[\begin{split}
\nu (L_{1}(x^*)^c\cap L_2(x^*)^c)
&= \nu ((L_1(x^*)\cup L_2(x^*))^c) \\
&\geq 1- \nu(L_1(x^*)) - \nu (L_2(x^*))> 0.    
\end{split}
\] Hence, $L_1(x^*)^c\cap L_2(x^*)^c\neq \emptyset$ and $x^*$ is not Pareto optimal; this contradicts the definition of a compromise.

Finally, remark that each agent is indifferent over all compromise solutions. Suppose not, so that there are two compromise solutions $y^*$ and $z^*$ with $y^*\succ_i z^*$ for some $i$. By Lemma \ref{lem:largeset}, $\nu(L_i(y^*))>\nu(L_i(z^*))$. But $\nu(L_1(y^*)) = \nu (L_2(y^*))$, so the Pareto optimality of $z^*$ leads to a contradiction.
\end{proof}

The following examples illustrate a regular problem, a non-regular problem (both with a single compromise), and a regular problem with two compromise solutions. Note that in a non-regular problem, the conclusion of Lemma \ref{lem:PO} according to which the measures of both lower contour sets are identical at a compromise option may not remain valid.

\begin{example}\label{ex:regular} Let $X=[0,1]$ be the policy space and assume that the two agents have single-peaked preferences over $X$, with $a_i$ denoting the ideal point of agent $i$ as illustrated by Figure \ref{fig:singlepeaked}. Figure \ref{fig:singlepeakedleft} represents the utility functions $u_i$ and Figure \ref{fig:singlepeakedright} the measure of their respective lower-contour sets, denoted $\nu(L_i(x))$ with $\nu$ being the Lebesgue measure. For each $x\in X$, Agent $1$'s utility function is $u_1(x)=1-|x - 1/10|$ with ideal point $a_1=1/10$ whereas
Agent 2 has ideal point $a_2=1/2$ and a utility function:
$$u_2(x) = \begin{cases} 1-2|1/2 - x| & \text{if } x \leq 0.5, \\ 1-|1/2 - x| & \text{if } x > 0.5. \end{cases}$$
Notice the difference between the utility function of each agent and the common cardinalization we have adopted through the measure of the agents' lower contour sets.

As preferences are not symmetric, the compromise $x^*$ is not the midpoint of $a_1$ and $a_2$. Indeed, the compromise $x^*= 0.375$ coincides with the maximum of the minimum of the measures of the lower-contour sets with $x^*\neq \frac{a_1+a_2}{2}=0.3$.

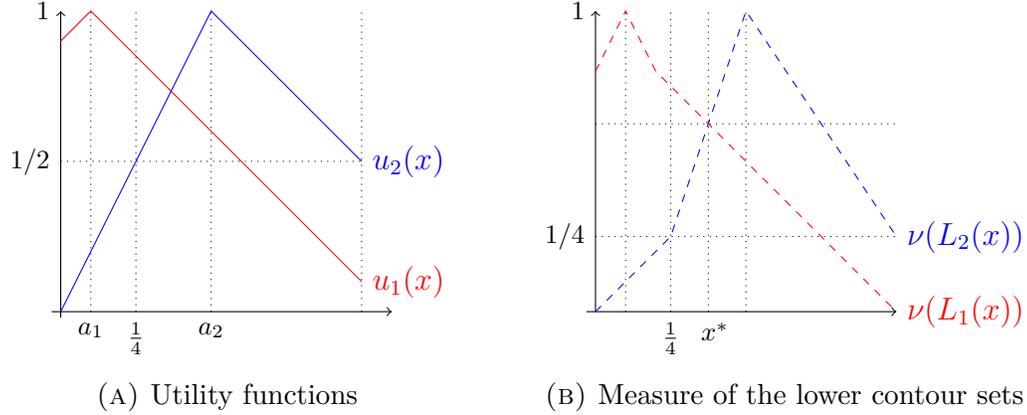
\begin{figure}
    \centering

        \begin{subfigure}[b]{0.45\textwidth}
     \centering
    \begin{tikzpicture}[scale=4]

        \coordinate (A1) at (0.1, 1); 
    \coordinate (A2) at (0.5, 1); 
    \coordinate (U1_0) at (0, 0.9); 
    \coordinate (U1_1) at (1, 0.1); 
    
    \draw[->] (-0.03,0) -- (1.1,0);
\draw[->] (0,-0.02) -- (0,1);

    \draw[red, thin] (U1_0) -- (A1) -- (U1_1) node[anchor= west] {$u_1(x)$};

                  \draw[thin,dotted] (0.5,0) -- (0.5,1) ;
                  \draw[thin,dotted] (0.1,0) -- (0.1,1) ;
   \draw[thin,dotted] (0,0.5) -- (1,0.5) ;
      \draw[thin,dotted] (1/4,0) -- (1/4,1) ;
      \draw[thin,dotted] (1,0) -- (1,1) ;

    \node[left] at (0,0.5) {\footnotesize $1/2$}; 
    \node[left] at (0,1) {\footnotesize $1$}; 
       \node[below] at (1/4,0) {\footnotesize $\frac{1}{4}$}; 
    \node[below] at (0.1,0) {\footnotesize $a_1$};
      \node[below] at (0.5,0) {\footnotesize $a_2$};

    \coordinate (U2_0) at (0,0);     \coordinate (U2_1) at (1, 0.5); 
    \draw[blue, thin] (U2_0) -- (A2) -- (U2_1) node[anchor= west] {$u_2(x)$};
    
\end{tikzpicture}
\caption{Utility functions} 
 \label{fig:singlepeakedleft}
\end{subfigure}
 \quad
    \begin{subfigure}[b]{0.45\textwidth}
    \centering
    \begin{tikzpicture}[scale=4]

\draw[->] (-0.03,0) -- (1,0);
\draw[->] (0,0) -- (0,1);
        \coordinate (A1) at (0.1, 1.0);     \coordinate (A2) at (0.5, 1.0);     \coordinate (X0) at (0, 0);        \coordinate (X1) at (1, 0);        \coordinate (Y1) at (0, 1);    
        \coordinate (I) at (0.375, 0.625);

            \coordinate (F1_0) at (0, 0.8);
    \coordinate (F1_4) at (0.2, 0.8);
    \coordinate (F1_1) at (1, 0);
    
    \draw[red, thin, dashed] (F1_0) -- (A1) -- (F1_4) -- (F1_1) node[anchor= west] { $\nu(L_1(x))$};

    \coordinate (F2_25_kink) at (0.25, 0.25);
    \coordinate (F2_1) at (1, 0.25);
    
                      \node[below] at (0.375,0) {\footnotesize $\mbox{ }\,x^*$};
                \node[left] at (0,1/4) {\footnotesize $1/4$};        \node[below] at (1/4,0) {\footnotesize $\frac{1}{4}$}; 
        \node[left] at (0,1) {\footnotesize $1$};
             
    \draw[blue, thin, dashed] (X0) -- (F2_25_kink) -- (A2) -- (F2_1) node[anchor=west] { $\nu(L_2(x))$};

        \draw[thin,dotted] (0.375,0) -- (0.375,1) ;
        \draw[thin,dotted] (0.5,0) -- (0.5,1) ;
        \draw[thin,dotted] (0.25,0) -- (0.25,1) ;

        \draw[thin,dotted] (0.1,0) -- (0.1,1) ;
       \draw[thin,dotted] (0,0.625) -- (1,0.625) ;
        \draw[thin,dotted] (0,1/4) -- (1,1/4) ;

\end{tikzpicture}
\caption{Measure of the lower contour sets}  
 \label{fig:singlepeakedright}
\end{subfigure}
    \caption{Compromise in a regular problem .}
   \label{fig:singlepeaked}
\end{figure}
\end{example}

\begin{example}
As in Example \ref{ex:regular}, let $X=[0,1]$ be the policy space. Both agents have piecewise linear utility functions, as illustrated in Figure~\ref{fig:non-regularexample}. Figure~\ref{fig:non-regularexampleleft} depicts the utility functions of agents 1 and 2, and Figure~\ref{fig:non-regularexampleright} depicts the measures of the lower contour sets for each agent. The utility of agent 1 (drawn in red) is piecewise linear with two segments: one that is strictly decreasing from $0$, where it takes the value $1/3$, to $1/2$, where it is $0$; and then strictly increasing from $1/2$ to $1$. The utility of agent 2 (drawn in blue) is also piecewise linear with two segments: one from $0$ to $2/3$ and the other from $2/3$ to $1$. In the case of agent 2, the utility is everywhere strictly decreasing. At $x=0$, 1's utility is $1/3$ and 2's is $1$. At $1$, 1's utility is $1$ and 2's is $0$. This fully describes the agents' utilities, which are, of course, continuous and feature thin indifference curves (each indifference curve is finite).

\begin{figure}
    \centering
        \begin{subfigure}[b]{0.45\textwidth}
        \centering
                \begin{tikzpicture}[scale=4]
            \draw[->] (-0.03,0) -- (1,0);
            \draw[->] (0,-0.03) -- (0,1);
            \node[below] at (0,0) {\footnotesize $0$};
            \node[below] at (2/3,0) {\footnotesize $2/3$};
            \node[below] at (1/2,0) {\footnotesize $1/2$};
            \node[below] at (1,0) {\footnotesize $1$};
            \draw[thin,dotted] (0,1/3)  node[left] {\footnotesize $1/3$} -- (1,1/3);
            \draw[thin,dotted] (1/2,0) -- (1/2,1);
            \draw[thin,dotted] (1,0) -- (1,1);
            \draw[thin,dotted] (0,1) node[left] {\footnotesize $1$} -- (1,1);
            \draw[thin,dotted] (2/3,0) -- (2/3,1);
            \draw[thin,red] (0,1/3) -- (1/2,0) -- (1,1) node[anchor=west] { $u_1(x)$};
            \draw[thin,blue] (0,1) -- (4/6,1/6) -- (1,0) node[anchor= west]  { $u_2(x)$};
        \end{tikzpicture}
        \caption{Utility functions}
        \label{fig:non-regularexampleleft}
    \end{subfigure}
    \quad         \begin{subfigure}[b]{0.45\textwidth}
        \centering
                \begin{tikzpicture}[scale=4]
            \draw[->] (-0.03,0) -- (1,0);
            \draw[->] (0,-0.03) -- (0,1);
            \node[below] at (0,0) {\footnotesize $x^*$};
            \node[below] at (2/3,0) {\footnotesize $2/3$};
            \node[below] at (1/2,0) {\footnotesize $1/2$};
            \node[below] at (1,0) {\footnotesize $1$};
                        \node[left] at (0,4/6) {\footnotesize $2/3$};
            \node[left] at (0,3/6) {\footnotesize $1/2$};
            \node[left] at (0,6/6) {\footnotesize $1$};
                        \draw[blue,thin,dashed] (0,1) -- (1,0) node[anchor=west] {$\nu(L_2(x))$};
                                    \draw[red,  thin, dashed] (0,4/6) -- (3/6,0);
                        \draw[red, thin, dashed] (3/6,0) -- (4/6,4/6);
                        \draw[red, thin,dashed] (4/6,4/6) -- (6/6,6/6) node[anchor=west] { $\nu(L_1(x))$};
            
                                    \draw[thin,dotted] (3/6,0) -- (3/6,6/6);
                        \draw[thin,dotted] (4/6,0) -- (4/6,6/6);
                        \draw[thin,dotted] (0,3/6) -- (6/6,3/6);
                        \draw[thin,dotted] (0,4/6) -- (6/6,4/6);
            \draw[thin,dotted] (6/6,0) -- (6/6,6/6);
            \draw[thin,dotted] (0,6/6) -- (6/6,6/6);
        \end{tikzpicture}
        \caption{Measure of the lower contour sets }
        \label{fig:non-regularexampleright}
    \end{subfigure}
    
    \caption{Compromise in a non-regular problem}
    \label{fig:non-regularexample}
\end{figure}
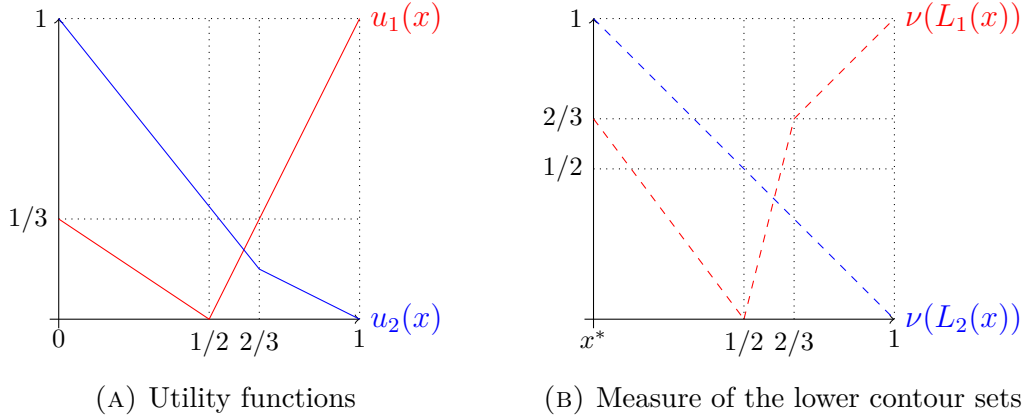
 The unique compromise point is $x^*=0$, which coincides with the maximum of the minimum of both curves. Remark that Lemma \ref{lem:PO} does not apply in non-regular problems, as at $x^*=0$, the measure of the lower contour sets is $2/3$ for agent 1 and $1$ for agent 2. Indeed, any other alternative $x\in [0,2/3)$ has a strictly smaller $\nu(L_1(x))$; and for any $x\in [2/3,1]$, we have $\nu(L_2(x))\leq 1/3<\nu(L_1(x^*))$. The problem in the example is non-regular because agent 1's utility is maximal at $x^*$ in the neighborhood $[0,2/3)$.
\end{example}

\begin{example} The third example features multiple compromise solutions.
Two agents are respectively endowed with the utility functions  $u_1(x) = |2x - 1|$ and $u_2(x) = 1 - |2x - 1|$ over $X=[0,1]$ as illustrated by Figure \ref{fig:twocompromisesexample}.
Here, Agent 1 prefers the extremes and dislikes the center, whereas Agent 2 prefers the center and dislikes the extremes, with both utility functions being piece-wise linear, as shown in Figure \ref{fig:twocompromisesexampleleft}. In this example, the Lebesgue measure of the lower contour set for each utility function coincides with the utility function itself.  Observe that there are two compromises: $x^*_1=1/4$ and $x^*_2=3/4$. They coincide with the maximum of the minimum of both $\nu(L_1(x))$ and $\nu(L_2(x))$, as shown in Figure \ref{fig:twocompromisesexampleright}. It is therefore a regular problem, as the agents' preferences are locally non-satiated at either of these two compromises. In accordance with Lemma \ref{lem:PO}, both agents are indifferent over both compromises with $\nu(L_1(x^*))=\nu(L_2(x^*))=1/2$ for $x^*=x^*_1,x^*_2$.

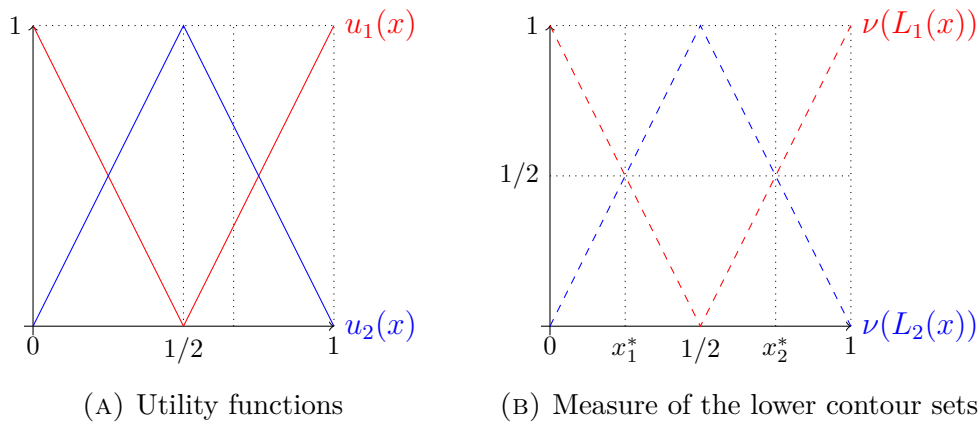
\begin{figure}
\centering
             \begin{subfigure}[b]{0.45\textwidth}
        \centering
                \begin{tikzpicture}[scale=4]
            \draw[->] (-0.03,0) -- (1,0);
            \draw[->] (0,-0.03) -- (0,1);
            \node[below] at (0,0) {\footnotesize $0$};
            \node[below] at (1/2,0) {\footnotesize $1/2$};
            \node[below] at (1,0) {\footnotesize $1$};
                       \draw[thin,dotted] (1/2,0) -- (1/2,1);
            \draw[thin,dotted] (1,0) -- (1,1);
            \draw[thin,dotted] (0,1) node[left] {\footnotesize $1$} -- (1,1);
            \draw[thin,dotted] (2/3,0) -- (2/3,1);
            \draw[thin,red] (0,1) -- (1/2,0) -- (1,1) node[anchor=west] { $u_1(x)$};
            \draw[thin,blue] (0,0) -- (1/2,1) -- (1,0) node[anchor= west]  { $u_2(x)$};
        \end{tikzpicture}
    \caption{Utility functions}\label{fig:twocompromisesexampleleft}

\end{subfigure}
        \begin{subfigure}[b]{0.45\textwidth}
        \centering
                \begin{tikzpicture}[scale=4]
            \draw[->] (-0.03,0) -- (1,0);
            \draw[->] (0,-0.03) -- (0,1);
            \node[below] at (0,0) {\footnotesize $0$};
            \node[below] at (1/2,0) {\footnotesize $1/2$};
            \node[below] at (1,0) {\footnotesize $1$};
               \node[below] at (1/4,0) {\footnotesize $x^*_1$};
                  \node[below] at (3/4,0) {\footnotesize $x^*_2$};
            \draw[thin,dotted] (0,1/2)  node[left] {\footnotesize $1/2$} -- (1,1/2);
            \draw[thin,dotted] (1/4,0) -- (1/4,1);
             \draw[thin,dotted] (3/4,0) -- (3/4,1);
            \draw[thin,dotted] (1,0) -- (1,1);
            \draw[thin,dotted] (0,1) node[left] {\footnotesize $1$} -- (1,1);
                        \draw[dashed,thin,red] (0,1) -- (1/2,0) -- (1,1) node[anchor=west] { $\nu(L_1(x))$};
            \draw[dashed, thin,blue] (0,0) -- (1/2,1) -- (1,0) node[anchor= west]  { $\nu(L_2(x))$};
        \end{tikzpicture}
         \caption{Measure of the lower contour sets}
   
    \label{fig:twocompromisesexampleright}  

\end{subfigure}    
    \caption{Two Compromises in a regular problem}
    \label{fig:twocompromisesexample}
\end{figure}
\end{example}

\section{Applications}\label{sec:applications}

This section reviews some applications of the compromise solution in different economic settings. We draw from models in political economy and behavioral economics to argue that our result is widely applicable to settings that receive a lot of attention in the literature.

\subsection{Political economy}\label{sec:PE}
We first apply our framework to the model of mandatory versus discretionary spending proposed by \cite{bowen2014mandatory}. There are two parties bargaining over the allocation space $(x_1, x_2, g) \in \mathbb{R}^3_+$, subject to the budget constraint $x_1 + x_2 + g \leq 1$, where $g$ denotes the level of public good provision and $x_i$ is the private good consumed by party $i$ ($i=1,2$). Thus $X=\{(x_1,x_2,g)\geq 0 : x_1+x_2+g = 1 \}$.

The utility for party $i$ is given by:
\begin{equation}\label{eq:publicgood}
    u_i(x_1, x_2, g) = x_i + \theta_i \log(g),
\end{equation}
where $\theta_i$ denotes party $i$'s preference parameter for the public good. We assume that $\theta_1 \geq \theta_2 > 0$ and $\theta_1 + \theta_2 < 1$. The latter condition ensures that the Pareto efficient level of public good provision ($g^* = \theta_1 + \theta_2$) lies within the interior of the budget set. Figure \ref{figure:utilityprivate} depicts an indifference curve for each party  and its lower contour set for the case with $1/2=\theta_1>\theta_2=1/3$.

The set of Pareto optimal allocations is defined by: \begin{equation}\label{eq:xx}
\{(x_1,x_2,g)\geq 0: x_1+x_2=1-(\theta_1+\theta_2) \text{ and } g^*=\theta_1+\theta_2\}.   
\end{equation}

To see why, remark that, for any arbitrary utility level $\bar{u}_2$, the optimization problem of agent 1 is equal to $\max_{x_1, x_2, g}  x_1 + \theta_1 \ln(g)$ subject to $ x_2 + \theta_2 \ln(g) = \bar{u}_2$, $x_1 + x_2 + g \leq 1$ and $x_1, x_2, g \geq 0$.
The first constraint refers to the utility of party 2 whereas the second one is as previously described. At the optimum, the second constraint is binding and thus $x_1+x_2+g=1$. This equality combined with the first one leads us to redefine the optimization problem as:
  $  \max_{g} \left[1 -g-  \bar{u}_2 + (\theta_1 +\theta_2)\ln(g) \right]
    \textrm{ subject to } g \geq 0. $ The FOCs prove that $\theta_1+\theta_2=g^*$, as wanted. This level of public good is the Samuelson level of the public good (it requires that the sum of the marginal benefits of the public good equals its marginal cost).

\begin{figure}
    \centering

\begin{tikzpicture}

                        \begin{axis}[
        name=plot1,
        width=7cm,
        height=7cm,
        axis lines=left,
        xmin=0, xmax=1.1,
        ymin=0, ymax=1.1,
        xlabel={$x_1$},
        ylabel={$x_2$},
                      restrict y to domain=0:1,
        restrict x to domain=0:1
    ]

                \fill[cyan!30] (axis cs:0,0) -- (axis cs:1,0) -- (axis cs:0,1) -- cycle;

                                \addplot[
            fill=white,
            draw=none,
            domain=0:1,
            samples=100
        ]
        {max(0, 1 - x - exp(-2*x))} \closedcycle;

                \addplot[thick, blue, domain=0:1, samples=100] {1 - x - exp(-2*x)};

                \addplot[black, thick, dashed] coordinates {(0,1) (1,0)};

              \node[blue!50!black] at (axis cs:0.3, 0.3) {$L_1$};

    \end{axis}

                        \begin{axis}[
        at={(plot1.outer east)}, anchor=outer west, xshift=1cm,         width=7cm,
        height=7cm,
        axis lines=left,
        xmin=0, xmax=1.1,
        ymin=0, ymax=1.1,
        xlabel={$x_1$},
        ylabel={$x_2$},
        restrict y to domain=0:1,
        restrict x to domain=0:1
    ]

                \fill[red!40] (axis cs:0,0) -- (axis cs:1,0) -- (axis cs:0,1) -- cycle;

        \addplot[
            fill=white,
            draw=none,
            variable=\t,
            domain=0:1,
            samples=100
        ]
        ({max(0, 1 - t - exp(-3*t))}, t) \closedcycle;

                \addplot[thick, red, variable=\t, domain=0:1, samples=100] ({1 - t - exp(-3*t)}, t);

                \addplot[black, thick, dashed] coordinates {(0,1) (1,0)};

        \node[red!50!black] at (axis cs:0.5, 0.3) {$L_2$};

    \end{axis}

\end{tikzpicture}
\caption{Utility of Party 1 and 2 over private consumption with $\theta_1=1/2$ and $\theta_2=1/3$. Lower contour sets are colored.}
\label{figure:utilityprivate}
\end{figure}
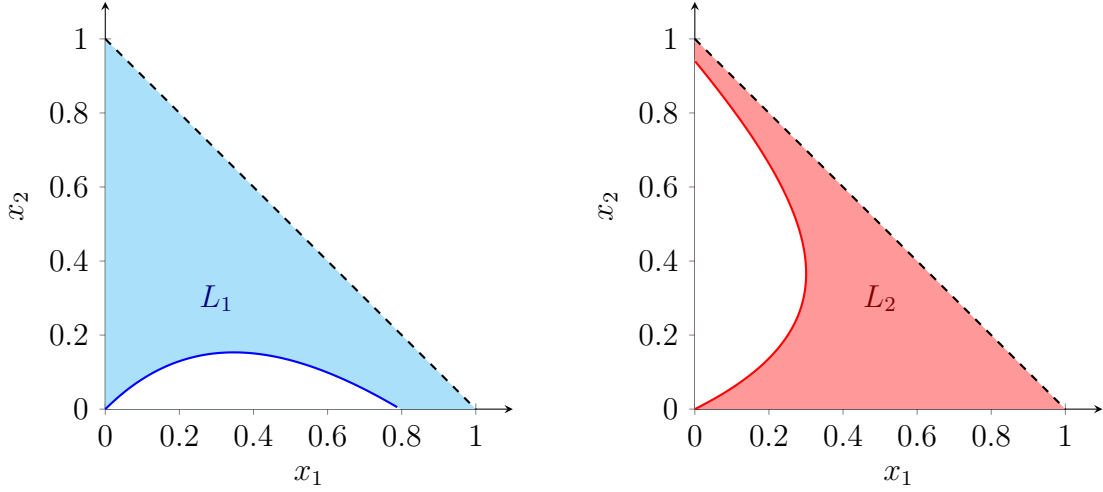

It turns out that there is a unique compromise allocation in this model, which can be explicitly expressed as a function of $\theta_1$ and $\theta_2$,  allowing to do some simple comparative statics.

\begin{proposition}\label{prop:PE}
    In the model of mandatory versus discretionary spending, for any $(\theta_1,\theta_2)$ there exists a unique compromise allocation $x^*=(x_1^*,x_2^*,g^*)$  in which $g^*=\theta_1+\theta_2$ while
    \begin{align*}
x^*_1 &= \frac{1}{g^*}\left( \theta_1(1 - g^*)+\theta_1\theta_2\log\left(\frac{\theta_2}{\theta_1}\right) \right) \\
x^*_2 &= \frac{1}{g^*}\left( \theta_2(1-g^*) +\theta_1\theta_2 \log\left(\frac{\theta_1}{\theta_2}\right) \right)
    \end{align*}
\end{proposition}

In the symmetric case in which $\theta_1 = \theta_2 = \theta$, the compromise allocation divides the available private budget equally:
$x_1^* = x_2^* = \frac{1 - 2\theta}{2}$. More generally, Proposition~\ref{prop:PE} implies that \[
\log(\frac{\theta_1}{\theta_2}) = \frac{x^*_2}{\theta_2}-\frac{x^*_1}{\theta_1}.
\] So when $\theta_1>\theta_2$ we must have $\frac{x^*_2}{\theta_2}>\frac{x^*_1}{\theta_1}$. When Party 1 cares more about the public good, its  consumption of the private good, normalized by his taste parameter, is lower than Party 2's.  In other words, at the compromise, Party 1 ``pays'' for higher public good provision through relatively lower private consumption (after adjusting for preferences).

\subsection{Private consumption with other-regarding preferences}\label{sec:fehr-schmidt}

As an application of the compromise solution, we consider a model of private consumption with other-regarding preferences. In particular, we consider one of the most widely applied models of social preferences and inequality aversion: \cite{fehr1999theory}. This model has been used to explain deviations from rational behavior in several settings, such as the ultimatum game. Indeed, in controlled lab experiments, subjects tend to reject unfair offers and be more generous than one would expect from a selfish and rational agent. Fairness is a central notion in the theory of \cite{fehr1999theory}, and thus arguably helps motivate the implementation of a compromise solution. 

The parametric nature of the model is particularly useful, as it allows us to derive closed-form solutions and obtain natural comparative statics. 

Within this section, we let $X=\{(x_1,x_2)\mid x_1,x_2\geq 0 \text{ and } x_1+x_2\leq 1\}$ be the outcome space where $x_1$ and $x_2$ respectively denote player 1 and player 2's monetary payoffs. This outcome space represents the division of a dollar between two players where ``waste'' (in the form of $x_1+x_2<1$) is allowed. Each player $i$, $i=1,2$, has a Fehr-Schmidt utility function with:
\begin{align*}
    U_i(x_i, x_j) &= x_i - \alpha_i \max\{x_j-x_i, 0\} - \beta_i \max\{x_i-x_j, 0\},
\end{align*}
where $i\neq j$, $ \beta_i\leq \alpha_i$, and $0< \beta_i <1$. Recall that $\alpha_i$ measures how much utility player $i$ loses when he obtains less than player $j$  (i.e., envy). In contrast,  $\beta_i$ measures how much utility player $i$ loses when he obtains more than player $j\neq i$ (i.e., guilt). Several works have attempted to estimate these parameters: one of the most recent works in this area, the meta-analysis by \cite{Nunnaripozzi2025}, finds that the average $\alpha$ is  0.533 and the average $\beta$ equals 0.326 relying on  structural estimation.

\subsubsection{Pareto Efficient allocations}
Remark first that any Pareto efficient allocation is such that $x_1+x_2=1$; that is, it lies on the feasibility frontier.\footnote{To see why, observe that for any allocation $(x_1,x_2)$ with $x_1+x_2<1$ we can construct an allocation $(x'_1,x'_2)=(x_1+\varepsilon,x_2+\varepsilon)$, with $\varepsilon>0$ small enough, where $x'_2-x'_1=x_2-x_1$. Thus, both players' utilities increase when shifting from $(x_1,x_2)$ to $(x'_1,x'_2)$, contradicting that $(x_1,x_2)$ is Pareto efficient. }

For simplicity, we parametrize any allocation in the budget line by Player 1's share $x$, where $x \in [0, 1]$ so that Player 2 receives $1-x$. This parametrization allows rewriting the utility functions as follows. 

If $x<1/2$, then:
\begin{align*}
    U_1(x) &= x - \alpha_1(1-2x)=(1+2\alpha_1)x-\alpha_1,\\
    U_2(x) &= (1-x) - \beta_2(1-2x)=(2\beta_2-1)x+(1-\beta_2),
\end{align*}
whereas when $x>1/2$, we obtain that
\begin{align*}
    U_1(x) &= x - \beta_1(2x-1)=(1-2\beta_1)x+\beta_1,\\
    U_2(x) &= 1-x - \alpha_2(2x-1)=-(1+2\alpha_2)x+(1+\alpha_2).
\end{align*}

As $\alpha_1>0$, the marginal utility of Player 1 w.r.t. $x$ is then positive if $x<1/2$. When $x>1/2$, its sign depends on the value of $\beta_1$: it is positive if $\beta_1<1/2$ and negative otherwise. Likewise, the marginal utility of Player 2 w.r.t. $x$ is negative if $x>1/2$ as $\alpha_2>0$. If $x<1/2$, its sign depends on the value of $\beta_2$: it is positive if $\beta_2>1/2$ and negative otherwise. We do not consider the case where some agent has $\beta_i=1/2$ as it generates a large set of indifferences.

Observing that an allocation is Pareto optimal if and only if the marginal utilities of the players have opposite signs (i.e., $\text{sgn}(U'_1(x)) \neq \text{sgn}(U'_2(x))$) allows us to derive the set of Pareto optimal allocations, as summarized in Table \ref{table:pofs}.

Note that the envy parameters ($\alpha_1, \alpha_2$) do not affect the set of Pareto allocations. This is because standard Fehr-Schmidt assumptions require $\alpha_i \geq 0$, ensuring that the marginal utility in the domain of disadvantageous inequality ($1+2\alpha_i$) is always strictly positive. Therefore, only the guilt parameters ($\beta_1, \beta_2$) determine the sign changes in marginal utility and, consequently, the Pareto optimal regions.

\begin{table}[ht]
    \centering
    \renewcommand{\arraystretch}{1.5}
    \setlength{\tabcolsep}{12pt}
    \begin{tabular}{cc|cc}
        \toprule
         & & \multicolumn{2}{c}{\textbf{Player 2's guilt}} \\
         & & $\beta_2 < 1/2$ & $\beta_2 > 1/2$ \\
        \midrule
        \multirow{2}{*}{\textbf{Player 1's guilt}}
        & $\beta_1 < 1/2$
        & $x \in [0, 1]$
        & $x \in [1/2, 1]$ \\
        \cmidrule{2-4}
        & $\beta_1 > 1/2$
        & $x \in [0, 1/2]$
        & $x = 1/2$ \\
        \bottomrule
    \end{tabular}
    \caption{Pareto optimal allocations in the Fehr-Schmidt model. \footnotesize All allocations lie on the feasibility frontier where $x_1+x_2=1$. Each allocation is parametrized by $x$, Player 1's monetary payoff.}
    \label{table:pofs}
\end{table}

\subsubsection{Indifference curves}

To determine the compromise solution $x^*$, we first determine the indifference curves for each player and then compute the measure of the upper-contour sets (where we rely on the fact that the measure of any upper-contour set is equal to the measure of $X$ minus the measure of the lower-contour set).

\medskip

\paragraph*{Players' indifference curves.}

Consider the utility function of Player 1. The derivation of 2's indifference curve is analogous. To simplify notation, we let $\alpha = \alpha_1$ and $\beta = \beta_1$. For any allocation $(x,y) \in X$, the utility of Player 1 is:
\[
    U(x,y) = x - \alpha\max\{y-x,0\} - \beta\max\{x-y,0\}.
\]
This simplifies to:
\begin{align*}
    U(x,y) &= (1+\alpha)x - \alpha y, &\quad \text{if } x < y \text{ (Envy)},\\
    U(x,y) &= (1-\beta)x + \beta y,  &\quad \text{if } x \ge y \text{ (Guilt)}.
\end{align*}

It follows that each indifference curve is piecewise linear. We denote the first part as the ``envy segment'' (dependent only on $\alpha$) and the second part as the ``guilt segment'' (dependent only on $\beta$). The curve is increasing with slope $\frac{1+\alpha}{\alpha}$ when $x < y$ (since $\alpha > 0$) and decreasing with slope $-\frac{1-\beta}{\beta}$ when $x \ge y$ (provided $\beta > 0$; it is vertical if $\beta=0$).

The Envy and Guilt segments intersect at the 45-degree line where $x=y$, creating a kink in the indifference curve. Specifically, for any utility level $\bar{u} \in [0,1/2]$, the allocation $(\bar{u}, \bar{u})$ lies on the indifference curve of level $\bar{u}$.

\subsubsection{Compromise and Inequality aversion}

We consider the outcome space $X = \{ (x_1, x_2) \in \mathbb{R}^2_+ \mid x_1 + x_2 \le 1 \}$. The Lebesgue measure of $X$, denoted $\nu(X)$, is $1/2$. 

\begin{theorem}\label{theo:fs} Assume that Players 1 and 2 have Fehr-Schmidt preferences over $X$ with parameters $(\alpha_1,\beta_1,\alpha_2,\beta_2)$. Suppose that $\beta_1,\beta_2<1/2$.  There exists a unique compromise alternative $x^*$:
\begin{enumerate}
    \item If $\beta_1=\beta_2$, then $x^*=1/2$.
\item If $\beta_1>\beta_2$, then $x^*<1/2$.
\item If $\beta_1< \beta_2$, then $x^*>1/2$.
\end{enumerate}
    \end{theorem}

The player with the larger guilt parameter obtains a lower share. Conditional on having a lower share, the more envious a player is, the harder it is to satisfy her.

\subsection{Spatial voting and facility location.} A popular model in mechanism design is the ``facility location'' problem in which a location in Euclidean space has to be chosen \citep{OWEN1998423}. The same model is very common in political science, where a ``spatial'' structure is given to the set of possible policies (see, for example, \cite{otto1970} and \cite{border1983}). We discuss a simple example, employing the facility location terminology.

Two agents are to share a common facility. All they care about is the facility's location, $x\in\Re^2$. Suppose that the set of feasible locations forms a rectangle $X$, endowed with the Lebesgue measure. Each agent has Euclidean preferences, with ideal points $i_1$ and $i_2$. An example is illustrated in Figure~\ref{fig:location}.

\begin{figure}
\begin{tikzpicture}[scale=4]

    \coordinate (A) at (0,0);
  \coordinate (B) at (2,0);
  \coordinate (C) at (2,1);
  \coordinate (D) at (0,1);

\coordinate (E) at (5/4,3/4);
  
    \coordinate (i1) at (1,1);
  \coordinate (i2) at ({1 + (sqrt(2)/4 + 1/4)/sqrt(2)}, {1 - (sqrt(2)/4 + 1/4)/sqrt(2)});

    \draw[line width=0.1mm] (A) -- (B) -- (C) -- (D) -- cycle;

    \begin{scope}
    \clip (A) rectangle (C);
    \fill[cyan!30] (i1) circle ({sqrt(2)/4});
    \draw[blue, thick] (i1) circle ({sqrt(2)/4});
           \draw[blue, thick,dashed] (i1) circle (.1);
                      \draw[blue, thick,dashed] (i1) circle (.175);
                     \draw[blue, thick,dashed] (i1) circle (.25);
  \end{scope}

    \fill[red!30] (i2) circle (1/4);
  \draw[red, thick] (i2) circle (1/4);

          \draw[red!60, thick,dashed] (i2) circle (.1);
            \draw[red!60, thick,dashed] (i2) circle (.15);
                  \draw[red!60, thick,dashed] (i2) circle (.2);

    \filldraw[black] (i1) circle (0.01) node[above left] { $i_1$};
  \filldraw[black] (i2) circle (0.01) node[below right] { $i_2$};

  \filldraw[black] (E) circle (0.01) node[anchor= west] { $x^*$};
  
\end{tikzpicture}
\caption{Facility location example.}\label{fig:location}
\end{figure}
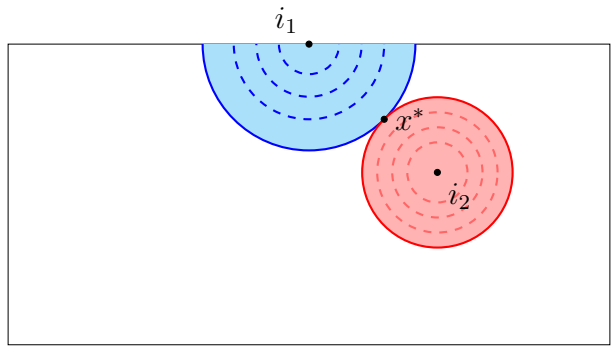

In the example, the measures of the lower and upper contour sets (at any point) add up to one. So the compromise point is achieved at a point that equalizes the measures of the upper contour sets: for graphical simplicity, we depict some dashed indifference curves and the upper contour sets at $x^*$ instead of the lower contour sets. The cyan half-disk centered at $i_1$, which is 1's upper contour set at $x^*$, has the same area as 2's upper contour set at $x^*$, represented by the red disk centered in $i_2$. Note that the radius of the larger circle is $\sqrt{2}$ times the radius of the smaller circle, and hence $x^*$ is closer to 2's ideal point than to 1's. The reason is that for a given radius, 1's upper contour set has $1/2$ the area of 2's. In equalizing the areas of the two upper contour sets, the compromise solution settles on a point closer to $i_2$ than to $i_1$.

\subsection{Agreeing on a lottery.} Here we revisit the setting from our story of the two Roman consuls in the introduction. Two agents have agreed to use a lottery to select one of the possible deterministic outcomes in $O=\{a,b,c\}$ but need to find an agreement on which lottery to use.\footnote{While using lotteries in allocation problems might be controversial, they are used in many important applications such as school choice, visas, etc. See \cite{bouacida2025rituals} for recent work on the social acceptability of lotteries and \cite{agranov2017stochastic} for experimental evidence on how individuals might want to rely on a randomization device to make their choice when indifferent. } We denote by $X=\Delta(O)$ the probability simplex over $O$. Assuming that players have expected utilities (i.e., linear utilities over $\Delta(O)$) implies that this example satisfies the conditions of Proposition \ref{prop:thinindiff} and thus that the problem is regular, as depicted by Figure \ref{fig:simplex}.

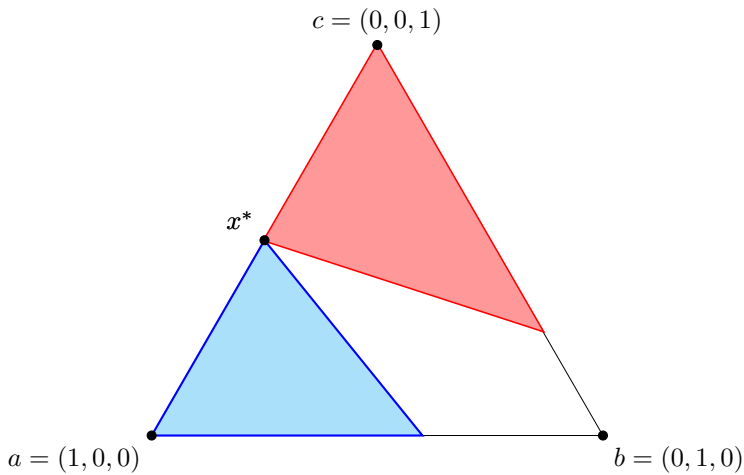
\begin{figure}
\centering
\begin{tikzpicture}[scale=6]

\coordinate (A) at (0,0);                   \coordinate (B) at (1,0);                   \coordinate (C) at (0.5,{sqrt(3)/2});       \coordinate (D) at (0.25,0.4330); 
    \coordinate (E) at (0.25,0.2330); 
  
\draw[line width=0.1mm] (A) -- (B) -- (C) -- cycle;

\node[below left] at (A) {\footnotesize $a=(1,0,0)$};
\node[below right] at (B) {\footnotesize $b=(0,1,0)$};
\node[above] at (C) {\footnotesize $c=(0,0,1)$};
\node[above left] at (D) {\footnotesize $x^*$};

\coordinate (P1) at (0, 0);                        \coordinate (P2) at (0.6, 0);                   \coordinate (P3) at (0.25, 0.4330);              
\fill[cyan!30] (P1) -- (P2) -- (P3) -- cycle;

\coordinate (Q1) at (0.5, 0.8660);                 \coordinate (Q2) at (0.25, 0.4330);                \coordinate (Q3) at (0.865, 0.2330);              \coordinate (S2) at (0.3, 0.5196);                \coordinate (S3) at (1.038, 0.2796);              \coordinate (T2) at (0.25, 0.4330);                \coordinate (T3) at (0.865, 0.2330);               
  \draw[red,very thick] (Q1) -- (Q2) -- (Q3) -- cycle;
    \draw[blue,thick] (P1) -- (P2) -- (P3) -- cycle;
  
\fill[red!40] (Q1) -- (Q2) -- (Q3) -- cycle;
 
  \filldraw[black] (D) circle (0.01) node[above left] {\footnotesize $x^*$};
    \filldraw[black] (A) circle (0.01) node[above left] {};
      \filldraw[black] (B) circle (0.01) node[above left] {};
        \filldraw[black] (C) circle (0.01) node[above left] {};
\end{tikzpicture}
\caption{Compromise in the simplex.}\label{fig:simplex}
\end{figure}

In the example, agent 1 and 2 have linear preferences over $\Delta(X)$ (that is, they are endowed with vN-M utilities). This means that their indifference curves in $\Delta(X)$ are straight parallel lines. Moreover, we consider that $u_1=(1,m,0)$ and $u_2=(0,n,1)$ with $0<m,n<1/2$. This implies that any Pareto efficient lottery assigns zero weight to $b$ and thus is located in the edge $ac$.\footnote{Indeed any lottery $p=(p_a,p_b,p_c)$ with $p_b>0$ is Pareto dominated by $p^\epsilon=(p_a+\epsilon/2,p_b-\epsilon,p_c+\epsilon/2)$ which is defined for $\epsilon>0$ small enough.} The measures of the lower and upper contour sets (at any point) add up to one. So, as in the facility location problem, the compromise point is achieved at a point that equalizes the measures of the upper contour sets: for graphical simplicity, we depict the upper contour sets instead of the lower contour sets. The cyan triangle represents 1's upper contour set at $x^*$ (whose ideal point is $a$), has the same area as 2's upper contour set at $x^*$ (whose ideal point is $c$), represented by the red triangle.

\section{Discussion \label{sec:discussion}}

Two-person implementation is hard. \cite{maskin99} and \cite{hurwicz1978construction} prove impossibility theorems for two-person Nash implementability of any Pareto optimal social choice function. The impossibility results of \cite{hurwicz1978construction} and \cite{maskin99} are very general. 
\cite{moorerepulloNash} and \cite{dutta1991necessary} provide necessary and sufficient conditions for an environment to ensure that implementation is feasible. These conditions can be difficult to check. \cite{sjostrom1991necessary} provides equivalent conditions along with an algorithm for checking them. 

In the case when $A$ is a finite-dimensional Euclidean space, \cite{moore1988subgame} show that, if their ``condition $E$'' holds, any social choice function that satisfies Maskin monotonicity and nonempty lower intersection can be Nash implemented. Our convex environments of Section~\ref{sec:convex}  satisfy their condition $E$, but (as we shall see) our solution may fail to be monotonic. Likewise, \cite{dutta1991necessary} provides a domain restriction where implementation is feasible via condition $\beta^*$ when $A$ is a compact finite-dimensional Euclidean space. Like us, they also consider an application where both agents choose a lottery in the simplex, but they impose several restrictions on  preferences over  lotteries (such as lack of unanimity) that we do not require. \cite{vartiainen2007nash} considers Nash implementation of bargaining solutions and proves that the exercise is not easy, as no Pareto-optimal and symmetric bargaining solutions can be Nash implemented.

\subsection{Nash implementation} We first comment on how our contribution relates to the literature on Nash implementation. First, we present an example to show that the compromise rule fails Maskin monotonicity. Since Maskin monotonicity is  a necessary condition for Nash implementation, the compromise rule is not Nash implementable. Then we generalize Maskin's negative result to the environments we have considered in our paper.

\begin{example}[Maskin monotonicity]
Let $X=[0,1]$ and suppose that preferences $\succeq_i$, $i=1,2$ are Euclidean with ideal points $x_1=0$ and $x_2=1$. The compromise is then $x^*=1/2$. Consider a preference profile $(\succeq'_1,\succeq'_2)$ in which $\succeq_1=\succeq'_1$ while $\succeq'_2$ is Euclidean with ideal point equal to $1/2$. Then it is clear that the lower contour set of agent 1 at $1/2$ is the same, while the lower contour set of agent 2 at $1/2$ is $X$, so it includes the lower contour set when 2's ideal point was $x_2=1$. Now, however, it is easy to see that $1/2$ is no longer a compromise. It is $1/3$.\footnote{Remark that $\nu(L_1(x))=1-x$ and $\nu(L_2(x))=1-2\vert x-\frac{1}{2}\vert$. The compromise option $x^*=1/3$ with $\nu(L_1(x^*))=\nu(L_2(x))=2/3$.} So Maskin monotonicity is violated.
\end{example}

Our second remark on Nash implementation is  much more general. There are well-known impossibility results regarding the Nash implementation of Pareto optimal social choice functions. These impossibility results impose certain conditions on the environments of interest. So it is natural to ask whether the impossibility results hold for the class of environments we have assumed (in particular, in environments in which condition $E$ of \cite{moore1988subgame} holds). The answer is negative if one considers continuous mechanisms, as we proceed to show in Proposition~\ref{prop:maskin}.

Let $R_i$ be the set of preference profiles of agent $i$. A social choice function $f:R_1\times R_2 \rightrightarrows X$ associates the set $f(\succeq_1,\succeq_2)\subseteq X$ to any pair of preference orderings. We say that $f$ is \df{Pareto optimal} if any $x\in f(\succeq_1,\succeq_2)$
is not Pareto dominated according to the preference profile $(\succeq_1,\succeq_2)$. A mechanism $(g,S_1,S_2)$ is \df{continuous} if each $S_i$ is a compact metric space and $g:S_1\times S_2\to X$ is continuous. 

\begin{proposition}\label{prop:maskin} Let $X=[0,1]$ and 
   $f$ be a Pareto optimal social choice function. If $(g,S_1,S_2)$ a continuous mechanism that implements $f$ in Nash equilibrium, then $f$ is dictatorial.
\end{proposition}

\begin{proof}Note that implementation requires the existence of a Nash equilibrium for any utility in the domain. We shall make use of the existence property. We mimic the steps in Maskin's proof.

Let $C_i(s_{-i})=\{g(\tilde s_i,s_{-i}) : \tilde s_i\in S_i \}$ be all the outcomes that $i$ can reach when $-i$ chooses $s_{-i}$. Then $C_i(s_{-i})$ is compact because $S_i$ is compact and $g(\cdot,s_{-i})$ is continuous.

\medskip

\textbf{Claim 1:} For all $s=(s_1,s_2)$, $C_1(s_{2})\cup C_2(s_{1})=X$.

Suppose, towards a contradiction, that there exist $s=(s_1,s_2)$ and $x\in X$ with $x\notin C_1(s_{2})\cup C_2(s_{1})$. Let $y=g(s)$ so that $y\in C_1(s_{2})\cup C_2(s_{1})$. Since  $C_1(s_{2})\cup C_2(s_{1})$ is closed, there is a small enough ball $B_x$ centered at $x$ so that 
$$B_x\subseteq [ C_1(s_{2})\cup C_2(s_{1}) ]^c.$$

So $y\in B_x^c$. We may now find preferences such that $x$ is the top in $X$ for both players and $y$ is the top in $B_x^c$:
Define for any $z\in X$, the function  \[v(z)=\begin{cases}
-|x-z| + \da & \text{ if } |x-z|\leq \da \\
0 & \text{ if } |x-z| > \da, \\
\end{cases}\]
and $w(z) = -|y-z|$ . Consider a utility $u_i(z)=v(z)+\ep w(z)$, choosing $\ep>0$ small enough so that $x$ is the global maximum of $u_i$. Then we have $y$ being a local maximum: a maximum of $u_i$ in $B_x^c$. Moreover, these utilities are in our domain because $v$ and $w$ are continuous and  $u_i$ is not constant in any open subset of $X$.

This means that $s$ is a Nash equilibrium, but its outcome is Pareto dominated, a contradiction. 

\medskip

\textbf{Claim 2:} If $x\in \bigcap_{s_2\in S_2} C_1(s_2)$ then there is $s^*_1\in S_1$ with $g(s^*_1,s_2)=x$ for all $s_2\in S_2$. Similarly for $x\in \bigcap_{s_1\in S_1} C_2(s_1)$. 

We may choose two utility functions $u_1(z)=-|x-z|$ and $u_2(z)=|x-z|$. Let $s^*$ be a Nash equilibrium (whose existence is assured by the hypothesis that $g$ Nash implements the social choice function). Then since $x\in \bigcap_{s_2\in S_2} C_1(s_2)$ we must have that $g(s^*)=x$. Now, if $g(s^*_1,s_2)\neq x$ for some $s_2$, then 2 would deviate to $s_2$, which would not be compatible with $s^*$ being a Nash equilibrium. So we conclude that $g(s^*_1,s_2)= x$ for all $s_2\in S_2$, as wanted.

\medskip

Now to finish the proof, let $x\in X$ and suppose that $x\notin \bigcap_{s_2\in S_2} C_1(s_2)$. Then       
there exists $\bar s_2\in S_2$ with $x\notin C_1(\bar s_2)$. But for any       
$s_1\in S_1$, $x\in C_1(\bar s_2)\cup C_2(s_1)$, by Claim 1. Thus  $x\in \bigcap_{s_1\in      
S_1} C_2(s_1)$. We conclude that                                               
$$X = \left(\bigcap_{s_2\in S_2} C_1(s_2)\right)\cup \left(\bigcap_{s_1\in S_1}      
C_2(s_1)\right).$$                                                             
                                                                               
Now we prove that, actually,                                                   
$X= \bigcap_{s_2\in S_2} C_1(s_2)$,                                               
or                                                                             
$X=\bigcap_{s_1\in S_1} C_2(s_1)$. To this end,                                   
suppose that there is $x,y\in X$, $x\neq y$, with                              
$x\in \bigcap_{s_2\in S_2} C_1(s_2)$ while                                        
$y\in \bigcap_{s_1\in S_1} C_2(s_1)$.                                             
                                                                              
By Claim 2 there exists $s^*_1$ with $x=g(s^*_1,s_2)$ for all $s_2$;           
and  $s^*_2$ with $y=g(s_1,s^*_2)$ for all $s_1$. But then we have             
$x=g(s^*_1,s^*_2)$ and $y=g(s^*_1,s^*_2)$, which is absurd.                    
                                                                               
Say that $X= \bigcap_{s_2\in S_2} C_1(s_2)$. Then by Claim 2, 1 is a dictator. A similar argument applies when $X= \bigcap_{s_1\in S_1} C_2(s_1)$, concluding the proof. \end{proof}

\subsection{Subgame perfect implementation }\label{sec:spnediscussion}
Our main result is on subgame-perfect implementation and is therefore closely related to the literature trying to understand which social choice functions are implementable in subgame-perfect equilibrium. One important difference with that literature is that our result provides a simple and (arguably) natural mechanism that implements the social choice function that we have focused on, the compromise rule. In line with the Jackson critique, one might say that the mechanism matters; thus, proposing a tractable and realistic mechanism is inherently valuable. Moreover, we work in an environment with two agents: the problem of subgame-perfect implementation with two agents remains open (\cite{vartiainen2007subgame} provides a complete characterization for three or more players). 

The first papers characterizing subgame perfect implementation are \cite{moore1988subgame} and \cite{abreu1990subgame}. \cite{moore1988subgame} provide a sufficient condition for a social choice function to be implementable. Their condition is not satisfied by our compromise function, as it requires 1) that the maximal elements of both agents are disjoint, and 2) the existence of a bad outcome.\footnote{Moreover, they prove the remarkable result that almost any social choice function can be implemented as the unique subgame perfect equilibrium of a carefully designed dynamic mechanism, provided that transfers are allowed. 
Our environment does not feature transfers, and their result is therefore not applicable. Subgame perfect implementation has been criticized by \cite{aghion2012subgame} and \cite{aghion2018role} for its reliance on complete information, but their results do not strictly speaking apply to our model with continuous strategy spaces.}
 \cite{abreu1990subgame} introduce a condition they call $\alpha$, which is necessary and almost sufficient for this sort of implementation. This condition bears the same relation to subgame perfect implementation as Maskin monotonicity to Nash implementation. Theorem \ref{th:implementation} shows that the compromise function is implementable; as a corollary, we conclude that this social choice function satisfies condition $\alpha$ in \cite{abreu1990subgame}. 

\subsection{Comparison with the finite setting}
For models with finitely many alternatives, the literature has defined compromise alternatives in analogous ways to our definition. An agent's utility from alternative $x$ may be taken to be the cardinality of that agent's lower contour set at $x$. The compromise set is then defined as the alternatives that maximize the welfare of the worst-off party; it contains either one or two alternatives (see \cite{hurwicz1999designing} and \cite{brams2001fallback} for further reference). In this context, \cite{anbarci1993noncooperative} proposed the VAOV mechanism that subgame perfectly implements a compromise alternative with potentially many stages. \cite{barbera2022compromising} propose three games with the same property and up to four stages, all requiring players to make a shortlist of alternatives.

\subsection{On bargaining solutions}
The literature on bargaining is vast. Its two main branches are axiomatic (see \cite{nash1950bargaining} and \cite{kalai1975other} for classic solutions and \cite{moulin1984implementing} for a game implementing the latter solution) and non-cooperative, or strategic, (as in \cite{rubinstein1982perfect}'s
 alternating offers model). We do not provide an axiomatization of the compromise rule; it's obviously important, and we plan to tackle an axiomatization in future work. 
 
 An important difference from the literature on non-cooperative bargaining is that we focus on full implementation, while the literature often elaborates on the equilibria of a given bargaining protocol. 
 
 Our model works directly with the set of outcomes.  In contrast, many contributions in bargaining are defined in utility space (under some specifications, the two spaces coincide). The compromise solution that we define is related to the Equal Area solution, which equates the measures of players' dissatisfaction in a setting with a disagreement point. \cite{anbarci1993noncooperative} proves that the equilibrium outcome of VAOV converges to this solution if the alternatives are uniformly distributed over the bargaining set as the number of
alternatives grows large. \cite{anbarci1994area} characterizes it. \cite{dasgupta2007bargaining} constructs a dynamic model where players have the power to destroy the feasible set, which gives rise to an equal area solution. Finally, \cite{li2023bargaining} shows how to characterize this solution by a suitable weakening of IIA jointly with Nash's classical axioms.

\section{Conclusion \label{sec:conclusion}}
We have considered a general class of social choice problems with two agents, a ``continuum'' set of outcomes, and thin indifference curves. The model encompasses many relevant applications in political economy, bargaining, and facility location, among others. We propose a natural notion of compromise that treats both agents symmetrically. 

Our model does not allow for unlimited monetary transfers. In fact, it encompasses models without ``money'' or quasilinear utility. We propose a no-transfers mechanism, the \emph{multimatum}. Our main result is that the multimatum fully implements the compromise solution that we have proposed. The literature has emphasized the ample possibilities for full implementation whenever transfers are available; yet, without transfers, and in spatial models of policy choice, our mechanism seems to be the first to reach this goal.

The main shortcoming of our approach is that it assumes two agents: our mechanism is necessarily a bilateral procedure that relies on offers and counter-offers. Extending the current techniques to many agents seems difficult. On the other hand, implementation with two agents is well-known to be challenging. Most characterizations of fully implementable social choice functions assume three or more agents. Our contribution with two agents is relevant for many applications in economics and political economy that naturally feature two parties.

An important component of our model is the measure $\nu$. The measure should be a natural way of assessing the size of a set of options and an accepted way of cardinalizing agents' utilities to make them comparable. It would be interesting to derive the measure from first principles. Perhaps from axiomatic principles that embody the idea of size and capture what makes some cardinalizations more acceptable than others. We have not pursued an axiomatic characterization in this paper, but it remains an interesting question for future research. 

A promising direction for future research is to investigate the mechanism's performance in an experimental setting. Social preferences may influence the second mover's response to the initial offer, a critical aspect of the current mechanism.

Another direction is to consider a version of this mechanism with an outside option, that plays the role of a status quo. The outside option is always among the possible choices for each player, in addition to the opponent's offers. This allows both players to stop the bargaining whenever no received offer is better than the initial conditions. Understanding the effect of the outside option on the Multimatum equilibria seems particularly relevant in several applications.

\section{Proofs}\label{sec:proofs}

\subsection{Preliminary results}

\begin{lemma}\label{lem:openpos} If $\emptyset \neq O\subseteq X$ is open, then $\nu(O)>0$.
  \end{lemma} \begin{proof}This property is equivalent to the full support of $\nu$.\end{proof}

\begin{lemma}\label{lem:largeset}If $x'\succ x$ then $\nu(L_{\succeq}(x'))>\nu(L_{\succeq}(x))$.
\end{lemma}

\begin{proof} By continuity of $\succeq$, $F =\{\tilde x\in X: x\succeq \tilde x \}$  and  $F'=\{\tilde x\in X: \tilde x\succeq x' \}$ are closed sets. Since $X$ is connected, $O=F^c\cap (F')^c = (F\cup F')^c\neq \emptyset$. Since $O$ is open, $\nu(O)>0$ by Lemma~\ref{lem:openpos}. But $O\subseteq L_{\succeq}(x')\setminus L_{\succeq}(x)$.
  \end{proof}

\begin{lemma}\label{lem:cont}
    The function $x\mapsto \nu(L_{\succeq}(x))$ is continuous.
\end{lemma}
\begin{proof}
 By the continuity of $\succeq$ and Debreu's theorem, $\succeq$ admits a continuous utility representation $u: X\to\Re$.
    
    Let $\{x^n\}$ be a sequence in $X$ and $x=\lim x^n$. Suppose first that $u(x^n)\leq u(x)$ for all $n$. Define a subsequence $x^{k_n}$ by: \[
    u(x^{k_n}) = \inf \{u(x^k):k\geq n \} \uparrow u(x)
    \] as $n\to\infty$. Then: \[
    L_{\succeq}(x^{k_n})\subseteq L_{\succeq}(x^k)\subseteq L(x),
    \] for all $k\geq n$. Then $L_{\succeq}(x) = I_{\succeq}(x)\cup \left( \bigcup_{n}^{\infty} L_{\succeq}(x^{k_n})\right)$. So, by the monotone continuity of probability measures, and the assumption of thin indifference curves,  $\nu(L_{\succeq}(x^{k_n}))\uparrow \nu(L_{\succeq}(x))$ as $n\to\infty$. Thus $\nu(L_{\succeq}(x^{k}))\to \nu(L_{\succeq}(x))$ as $k\to\infty$.
    
A similar squeezing argument from above yields that $\nu(L_{\succeq}(x^{k}))\to \nu(L_{\succeq}(x))$ when $u(x^k)\geq u(x)$ for all $k$. Now, in the general case, we can partition the sequence $x^n$ into two, one with  $u(x^n)\geq u(x)$ and one with $u(x^n)\leq u(x)$. Since $\nu(L_{\succeq}(x^{n}))\to \nu(L_{\succeq}(x))$ for both, the conclusion of the lemma holds. 
\end{proof}

\begin{lemma} \label{lem: ubs}
 For every $x\in X$ and $\ep>0$, there is a closed set $L_{\succeq}^\ep(x)\subseteq L_{\succeq}(x)$ from which $x$ is the uniquely best outcome and has the property that $\nu( L_{\succeq}^\ep(x))\geq \nu(L_{\succeq}(x)) - \ep$. \end{lemma}

\begin{proof}Under the assumption of thin indifference curves, $I_\succeq(x)\cap L_{\succeq}(x)$ has measure zero. So there is an open set $O$ with $\nu(O)<\ep$ and $I_\succeq(x)\cap L_{\succeq}(x)\subseteq O$. Letting \[L_{\succeq}^\ep(x) = \{x\}\cup (L_{\succeq}(x)\setminus O),
  \]
  concludes the proof.\end{proof}

We shall make use of the following lemma, which utilizes standard ideas in optimal transport.

\begin{lemma}\label{lem:villani}
  There exists a probability measure $\mu$ on $[0,1]^2$ with uniform marginals {$\la$} such that $\mu(\UPS)=1$ if and only if, for all Borel sets $E\subseteq [0,1]$,
  \[
\la(E)\leq \la(\UPS(E)),
  \]
where 
\[
\UPS(E) = \{v_2\in [0,1] \colon \exists v_1\in E \text{ s.t.\ } (v_1,v_2)\in \UPS \}.
\]
\end{lemma}

\begin{remark}\label{rmk:projection}
A classical problem in descriptive set theory is that a projection like $\UPS(E)$ may not be a Borel set (despite $(E\times [0,1])\cap \UPS$ being one). The projection is, however, analytic and therefore universally measurable (which makes $\la(\UPS(E))$ meaningful). 
\end{remark}

\begin{proof}
First, we prove necessity. Let $\mu$ be a probability measure on $[0,1]^2$ with marginals $\la$ such that $\mu(\UPS)=1$. Let $E\subseteq [0,1]$ be a Borel set. Since $\mu(\UPS)=1$ we have that
\[
\la(E) = \mu(E\times [0,1]) = \mu((E\times [0,1])\cap \UPS) \leq \mu([0,1]\times \UPS(E))
\]
because, if $(v_1,v_2)\in (E\times [0,1])\cap \UPS$ then $(v_1,v_2)\in \UPS$ while $v_1\in E$, and therefore $v_2\in \UPS(E)$.
Necessity follows because $\mu([0,1]\times \UPS(E)) = \la(\UPS(E))$ (recall Remark~\ref{rmk:projection}). 

Next, we turn to sufficiency. The proof follows from optimal transport duality when we consider a cost function that takes the value 1 outside of $\UPS$ and $0$ inside. Working with the dual, we show that there exists a zero-cost Kantorovich transport plan. This means that it must be supported on $\UPS$. 

Let $O = [0,1]^2 \setminus \UPS$, which is open as $\UPS$ is compact. Let $c(v_1,v_2) = \one_{O}(v_1,v_2)$. Observe that $c$ is lower-semicontinuous, as the complement of $\UPS$ is open. This puts us in the hypothesis of the main duality theorem of optimal transport (Theorem 1.3 in \cite{villani2003topics}).\footnote{In fact, our proof borrows heavily from Villani's proof of Strassen's theorem. See Appendix 1.4 in \cite{villani2003topics}.}

Consider the optimal transport problem:
\[
\min_{\mu \in \Pi(\la,\la)} \int_{[0,1]^2} c(v_1,v_2) \diff \mu(v_1,v_2),
\]
where $\Pi(\la,\la)$ is the set of all Borel probability measures on $[0,1]^2$ with marginals $\la$. 

The dual problem is 
\[
\sup_{(\phi, \psi) \in \Phi_c} \left( \int_{[0,1]} \phi(v_1) \diff \la(v_1) + \int_{[0,1]} \psi(v_2) \diff \la(v_2) \right),
\]
where
\[
\Phi_c = \{ (\phi, \psi) \in L^1([0,1]) \times L^1([0,1]) : \phi(v_1) + \psi(v_2) \le c(v_1,v_2)\,\, \la\times\la-\text{a.e.}\}.
\]
By Theorem 1.3 in \cite{villani2003topics}, we may without loss restrict attention to bounded and continuous elements of $\Phi_c$ on which the inequality $\phi+\psi\leq c$ holds pointwise.

Moreover, by standard arguments in optimal transport theory, given that $c(v_1,v_2)\in [0,1]$,  we may restrict attention in the dual to  continuous functions $(\phi,\psi)\in \Phi_c$ with $\phi(v_1)\in [0,1]$ and $\psi(v_2)\in [-1,0]$ for all $(v_1,v_2)\in [0,1]^2$. For completeness, here we spell out these arguments for the special case needed to prove our result.

First, following Remark 1.12 in \cite{villani2003topics}, in studying the dual, we may restrict attention to $c$-concave functions. In particular, consider the $c$-conjugates of any continuous $(\phi,\psi)\in\Phi_c$: 
\begin{align*}
    \psi^c(v_2) & = \inf \{ c(v_1,v_2) - \phi(v_1) \colon v_1\in [0,1] \} \\
    \phi^{cc}(v_1) & = \inf \{ c(v_1,v_2) - \psi^c(v_2) \colon v_2\in [0,1]\}.
\end{align*}
These are well-defined (they take finite values) given that we have taken $\phi$ and $\psi$ to be bounded. Note that $\psi^{c}$ is bounded, and since $\phi^{cc}(v_1) \leq c(v_1,v_2) - \psi^c(v_2)$, so is $\phi^{cc}$; therefore, $(\phi^{cc}, \psi^c)\in \Phi_c$. Observe that $\psi(v_2)\leq c(v_1,v_2) - \phi(v_1)$ for all $v_2$ implies that $\psi(v_2)\leq \psi^c(v_2)$. Similarly,   $\phi(v_1) \leq c(v_1,v_2) - \psi^c(v_1)$ implies that $\phi(v_1)\leq \phi^{cc}(v_1)$. Any time we entertain a pair $(\phi,\psi)\in \Phi_c$ in the solution to the dual, we can do no worse by instead switching to  $(\phi^{cc},\psi^{c})\in \Phi_c$. 

The next step is adapted from Remark 1.13 in \cite{villani2003topics}.  Fix $v_1,v'_1$. For any $\ep>0$ there is $\bar v_2$ such that 
\[
\phi^{cc}(v_1) > c(v_1,\bar v_2) - \psi^c(\bar v_2) -\ep
\]
while 
\[
\phi^{cc}(v'_1) \leq  c(v'_1,\bar v_2) - \psi^c(\bar v_2)
\]
Thus,
\[
\phi^{cc}(v'_1) -\phi^{cc}(v_1) \leq  c(v'_1,\bar v_2) - \psi^c(\bar v_2) - [ c(v_1,\bar v_2) - \psi^c(\bar v_2) -\ep]
< c(v'_1,\bar v_2) - c(v_1,\bar v_2) + \ep.
\]
Since $c$ only takes the values $0$ and $1$, and letting $\ep\to 0$, we obtain that 
\[
\phi^{cc}(v'_1) -\phi^{cc}(v_1) \leq 1
\] for all $v_1,v'_1$. Let $k=1-\sup \phi^{cc}$, so that $\phi^{cc}+k\leq 1$. Since $\phi^{cc}(v'_1) -\phi^{cc}(v_1) \leq 1$ for all $v_1,v'_1$, we conclude that $\phi^{cc}+k\in [0,1]$.

Its conjugate $\hat \psi$ satisfies that
\begin{align*}
  -1 \leq \inf\{0 - (\phi^{cc}(v_1) + k) \colon v_1\in [0,1] \}
& \leq \hat \psi(v_2) = \inf\{c(v_1,v_2) - (\phi^{cc}(v_1) + k) \colon v_1\in [0,1] \} \\
& \leq \inf\{1 - (\phi^{cc}(v_1) + k) \colon v_1\in [0,1] \} \leq 0  
    \end{align*}

    In conclusion, these conjugate transformations are measurable and bounded and can only improve on any original candidate solution. We may focus our attention on the bounded, measurable pair $(\phi,\psi)$ obtained from an initial pair $(\phi',\psi')\in \Phi_c$, with $\phi(v_1)\in [0,1]$ and $\psi(v_2)\in [-1,0]$ for all $(v_1,v_2)\in [0,1]^2$. Let $\Phi^*_c$  denote this set.

For any  $(\phi,\psi)\in \Phi^*_c$ and for any $s \in [0,1]$, define the upper level sets:
\[ A_s = \{ v_1  \colon \phi(v_1) > s \} \quad \text{and} \quad B_s = \{ v_2 \colon -\psi(v_2) > s \} \]

Using the ``layer-cake'' representation for non-negative functions (see \cite{villani2003topics} again), since $\phi$ has range $[0,1]$, we may represent it as
\[
\phi(v_1) = \int_0^1 \one_{\phi(v_1)>s} \diff s.
\]
So Fubini yields
\[
\int_{[0,1]} \phi \diff \la = \int_0^1 \la(A_s) \diff s.
\] Similarly for the function $-\psi$ which takes values in $[0,1]$.

This means that we may write the objective of the dual as:
\[
    \int_{[0,1]} \phi \diff \la + \int_{[0,1]} \psi \diff \la = \int_0^1 \left( \la(A_s) - \la(B_s) \right) \diff s.
\]

Consider $v_1\in A_s$. For any $v_2\in [0,1]$ with $(v_1,v_2)\in \UPS$, we have that $c(v_1,v_2)=0$ and hence that
\[
s<\phi(v_1)\leq c(v_1,v_2) - \psi(v_2) = -\psi(v_2). 
\] Thus $v_2\in B_s$. Consequently, $\UPS(A_s)\subseteq B_s$. Since $A_s$ is a Borel set, the hypothesis that $\la(A_s)\leq \la(\UPS(A_s))$ means that $\la(A_s)\leq \la(B_s)$. By duality, conclude that 
\begin{align*}
  \min_{\mu \in \Pi(\la,\la)} [1-\mu(\UPS)] & = 
  \min_{\mu \in \Pi(\la,\la)} \int_{[0,1]^2} \one_{[0,1]^2\setminus \UPS}(v_1,v_2) \diff \mu(v_1,v_2) \\
  & = \sup_{(\phi,\psi)\in\Phi^*_c}  \int_{[0,1]} \phi \diff \la + \int_{[0,1]} \psi \diff \la\leq 0,    
\end{align*}
and therefore that
\[
\max_{\mu \in \Pi(\la,\la)} \mu(\UPS)\geq 1.
\]
But, of course, for any probability measure $\mu$, $\mu(\UPS)\leq 1$. So the existence of optimal transportation maps gives us $\mu\in \Pi(\la,\la)$ and $\mu(\UPS)=1$.
\end{proof}

\subsection{Proof of Theorem~\ref{thm:commonnu}}

Notation: When $Z$ is a compact metric space, let $C(Z)$ denote the space of continuous functions $f:Z\to\Re$ endowed with the sup norm. This is a Banach space. Let $\la$ denote the Lebesgue measure on $[0,1]$.

First, we shall prove the second statement in the theorem.  Fix $i \in \{1,2\}$. We want to find a Borel probability measure $\nu_i$ on $X$ such that the push-forward $(u_i)_*\nu_i = \la$. Consider the operator $T_i : C([0,1]) \to C(X)$ defined by $T_i(f) = f \circ u_i$. Because $u_i$ is continuous, $T_i$ maps continuous functions to continuous functions. Since $u_i$ is surjective onto $[0,1]$, $\norm{f \circ u_i}_\infty = \norm{ f }_\infty$ for all $f\in C([0,1])$.

Let $M_i = T_i(C([0,1]))$. Recall that $C(X)$ is a Banach space and note that $M_i$ is a (closed) linear subspace of $C(X)$. We define a linear functional $\Lambda_i : M_i \to \Re$ by integrating with respect to Lebesgue measure:
\[ \Lambda_i(f \circ u_i) = \int_0^1 f(t) \diff \la(t). \] The functional is well defined because $u_i$ is surjective (thus $f \circ u_i=g \circ u_i$ implies $f=g$ for $f,g\in C([0,1])$). Given that $\norm{f \circ u_i}_\infty = \norm{ f }_\infty$ for all $f\in C([0,1])$, the norm of the functional $\Lambda$ is exactly 1. Furthermore, since $\one_X = T_i(\one_{[0,1]})$, we have $\Lambda_i(\one_X) = \int_0^1 1 \diff \lambda = 1$.

By the Hahn-Banach Extension Theorem (see, for example, Theorem 3.3 in \cite{rudin1991}), there exists a linear functional $\tilde{\Lambda}_i : C(X) \to \R$ extending $\Lambda_i$ such that $\|\tilde{\Lambda}_i\| = 1$. Strictly speaking, Hahn-Banach says that $\abss{\Lambda_i(f)}\leq \norm{f}_\infty$ for all $f\in M_i$ implies that its extension satisfies  $\abss{\tilde \Lambda_i(f)}\leq \norm{f}_\infty$ for all $f\in C(X)$. Since
\[\begin{split}
  \sup \{\abss{\tilde \Lambda_i(f)} : \norm{f}_{\infty}\leq 1\} & \geq
\sup \{\abss{\tilde \Lambda_i(f)} : f\in M_i \text{ and } \norm{f}_{\infty}\leq 1\} \\  & = 
\sup \{\abss{\Lambda_i(f)} : f\in M_i \text{ and } \norm{f}_{\infty}\leq 1\} = 1,\end{split}\] we conclude that $\norm{\tilde \Lambda_i}=1$ 

Next, we claim that  $\tilde{\Lambda}_i$ is a positive linear functional. We must show that for any $g \in C(X)$ with $g \ge 0$, we have $\tilde{\Lambda}_i(g) \ge 0$. By linearity and scaling, it is sufficient to prove this for any $g \in C(X)$ such that $0 \le g(x) \le 1$ for all $x \in X$. For such a function $g$, we have:
\[ 0 \le \one_X(x) - g(x) \le 1 \quad \text{for all } x \in X. \]
This implies that the  norm of the difference is bounded by 1: $\norm{\one_X - g}_\infty \le 1.$

Applying the functional $\tilde{\Lambda}_i$ to this difference, and using the fact that $\norm{\tilde{\Lambda}_i} = 1$, we obtain:
\[ \tilde{\Lambda}_i(\one_X - g) \le \norm{\tilde{\Lambda}_i} \cdot \norm{\one_X - g}_\infty \le 1 \cdot 1 = 1. \]
By the linearity of $\tilde{\Lambda}_i$, this gives:
\[ \tilde{\Lambda}_i(\one_X) - \tilde{\Lambda}_i(g) \le 1. \]
Since $\tilde{\Lambda}_i$ extends $\Lambda_i$, we know $\tilde{\Lambda}_i(\one_X) = \Lambda_i(\one_X) = 1$. So,
\[ 1 - \tilde{\Lambda}_i(g) \le 1 \implies \tilde{\Lambda}_i(g) \ge 0.
\]
Thus, $\tilde{\Lambda}_i$ is a positive linear functional on $C(X)$.

By the Riesz representation theorem (see \cite{Folland1999} Theorem 7.2 or what \cite{Aliprantis2006} term the Riesz-Markov-Kakutani Representation (Theorem 14.12)), there exists a unique regular Borel probability measure $\nu_i$ on $X$ such that for all $g \in C(X)$, 
\[ \tilde{\Lambda}_i(g) = \int_X g \diff \nu_i. \]
Turn now to the relevant push-forward measure $(u_i)_*\nu_i$. For $f \in C([0,1])$,
\[ \int_{[0,1]} f \diff ((u_i)_*\nu_i) = \int_X (f \circ u_i) \diff \nu_i = \tilde{\Lambda}_i(f \circ u_i) = \int_0^1 f \diff \la, \]
as $f \circ u_i\in M_i$, and $\tilde \Lambda_i$ is an extension of $\Lambda_i$. Since this holds for all continuous $f:[0,1]\to\Re$, it follows that $(u_i)_*\nu_i = \la$, proving the second statement in the theorem.

To prove the first statement in the theorem, we evaluate the measure $\nu_i$ on the lower contour set $L_i(x)$. By definition, $L_i(x) = u_i^{-1}([0, u_i(x)])$. Therefore:
\[ \nu_i(L_i(x)) = \nu_i\Big(u_i^{-1}([0, u_i(x)])\Big) = ((u_i)_*\nu_i)\Big([0, u_i(x)]\Big) = \la\Big([0, u_i(x)]\Big) = u_i(x). \]

This completes the proof of Statements 1 and 2. We turn to Statement 3.

Let $U : X \to [0,1]^2$ be the joint utility mapping defined by $U(x) = (u_1(x), u_2(x))$. Let $\UPS = U(X)$ be the image of $X$. Since $X$ is compact and $U$ is continuous, $\UPS$ is a compact subset of $[0,1]^2$. 

By Lemma~\ref{lem:villani}, there exists a probability measure $\mu$ on $[0,1]^2$ that is supported on $\UPS$ ($\mu(\UPS)=1$) if and only if the condition in the theorem is satisfied.

From $\mu$, we shall now obtain a probability measure $\nu$ on $X$. The idea is analogous to our proof of the second statement in the theorem, but applied to $U$ instead of $u_i$. We obtain $\nu$ by an application of (a version of) the Riesz representation theorem.

Consider the ``pullback'' $U^* : C(\UPS) \to C(X)$ given by $U^*(f) = f \circ U$. Because $U$ is surjective onto $\UPS$, $\norm{U^*(f)}_{\infty}=\norm{f}_{\infty}$ for all $f\in C(\UPS)$. Define the linear subspace $M = U^*(C(\UPS))$ of the Banach space $C(X)$, and define the bounded positive linear functional $\Lambda : M \to \R$ by
\[ \Lambda(f \circ U) = \int_\UPS f \diff \mu. \]
As before, $\norm{\Lambda} = \Lambda(\one_X) = 1$. By the Hahn-Banach Theorem, we extend $\Lambda$ to $\tilde{\Lambda}$ on all of $C(X)$ without increasing its norm. The functional $\tilde{\Lambda}$ is, again, positive. By Riesz-Markov-Kakutani, there exists a regular Borel probability measure $\nu$ on $X$ for which $\tilde \Lambda(f) = \int_X f(x) \diff \nu(x).$ 

For any $f \in C(\UPS)$, we have:
\[ \int_\UPS f \diff(U_*\nu) = \int_X (f \circ U) \diff \nu = \tilde{\Lambda}(f \circ U) = \int_\UPS f \diff \mu. \]
Thus, $U_*\nu = \mu$. Since $\mu$ has uniform marginals on $[0,1]$, it immediately follows that $(u_1)_*\nu = \la$ and $(u_2)_*\nu = \la$. By the same logic as in our proof of Statement 2, $\nu(L_i(x)) = u_i(x)$ for both $i \in \{1,2\}$ simultaneously. 

Finally, we prove the necessity statement in the theorem. Suppose that a common measure $\nu$ exists on $X$. Define $\mu$ as the push-forward measure $\mu = U_*\nu$ on $\UPS$. Since $U(X) = \UPS$, $\mu$ is supported entirely on $\UPS$. Because $(u_1)_*\nu = \la$ and $(u_2)_*\nu = \la$, the measure $\mu$ has uniform marginals. By the converse statement in Lemma~\ref{lem:villani}, the existence of such a $\mu$ implies that the  condition $\la(E) \le \la(\UPS(E))$ must hold for all Borel sets $E$.  \qed

 \subsection{Proof of Proposition~\ref{prop:thinindiff}}\label{sec:pfpropTIC}
Let $y\in I_{\succeq}(x)$, so $x\sim y$. By local strictness, for any neighborhood $N_y$ of $y$ and $N_x$ of $x$ there exists $y'\in N_y$ and $x'\in N_x$ with $x'\succ y'$.

First, if $x\succeq x'$, then $y\succeq x'\succ y'$. Thus $I^c_{\succeq}(x)\cap N_y\neq \os$. Second, if $x'\succ x$ then we may choose $\la\in (0,1)$ such that $\la y + (1-\la) x'\in N_y$. By explicit convexity, since $x'\succ x\sim y$, $\la y + (1-\la) x'\succ y$. Hence, again, $I^c_{\succeq}(x)\cap N_y\neq \os$. Hence, the complement of the indifference set $I^c_{\succeq}(x)$ is dense.

Under our assumptions, we have that $I_{\succeq}(x)=\partial U_{\succeq}(x)\coloneqq U_{\succeq}(x)\setminus U^o_{\succeq}(x)$ where $\partial U$ and $U^o$ respectively denote the boundary and interior of the set $U$. The sets $U_{\succeq}(x)$ and $U^o_{\succeq}(x)$ are convex. Indeed, $U_{\succeq}(x)$ is closed. If its interior is empty, then we must have $I_{\succeq}(x)=U_{\succeq}(x)= \partial U_{\succeq}$ by continuity, as $y\succ x$ would mean that $z\succ x$ for all $z$ in a (relative) neighborhood of $y$. If its interior is non-empty, then for any $y\sim x$ we may let $x'\in U^o_{\succeq}(x)$ satisfy $x'\succ y$ (such $x'$ exists since we proved that $I^c_{\succeq}(x)$ is dense and therefore must meet any neighborhood of an interior point of $U_{\succeq}(x)$) then $\la y +(1-\la) x'\succ y\sim x$ for any $\la\in (0,1)$. Thus $y\in \partial U_{\succeq}$. Either way we see that 
$I_{\succeq}(x)=\partial U_{\succeq}(x)$.

To prove that $I_{\succeq}(x)$ has measure zero, first assume that $U_{\succeq}(x)$ has a non-empty interior.
 For simplicity, suppose that we translate $U_{\succeq}(x)$ so that $0$ is in its interior. Such a translation does not affect the argument that the set $I_{\succeq}(x)$ has Lebesgue measure zero. For each $y\in I_{\succeq}(x)$, $z=(1-\ep) y\in U^o_{\succeq}(x)$ as $U(x)$ is convex and $0$ is an interior point. Thus $y=\frac{1}{1-\ep}z$. We conclude that $I_{\succeq}(x)\subseteq \frac{1}{1-\ep} U^o(x)$. Actually, $I_{\succeq}(x)\subseteq \frac{1}{1-\ep} U^o(x)\setminus U^o(x)$. Thus, if $\la$ denotes the Lebesgue measure, we obtain: \[
\la(I_{\succeq}(x))\leq \la\left( \frac{1}{1-\ep} U^o_{\succeq}(x)\setminus U^o_{\succeq}(x)\right) = [\frac{1}{1-\ep}-1]\la (U^o_{\succeq}(x))\to 0,
    \] as $\ep\to 0$. Here we have used the homogeneity property of scaling for Lebesgue measure. Since $\nu$ is absolutely continuous with respect to $\la$, the desired result follows.

    Finally, consider the case when $U_{\succeq}(x)$ has an empty interior. This case is standard and well-known: the affine hull of  $U_{\succeq}(x)$ must have dimension strictly less than $d$ with $X\subseteq \Re^d$ by assumption. As a consequence, it has Lebesgue measure zero. Since $\nu$ is absolutely continuous with respect to Lebesgue, $I_\succeq(x)\subseteq U_\succeq(x)$ also has $\nu$-measure zero. \qed

\subsection{Proof of Theorem~\ref{th:implementation}}\label{sec:proof}

\begin{lemma}\label{lem:2opt} Let $A_1\subseteq X$ be closed and non-empty. There is a strategy by Player 1 in which $A_1$ is offered in the first stage of the game, and Player 1 chooses optimally from each potential counterproposal by Player 2. Player 2 has a best response to this strategy by Player 1.
\end{lemma}

\begin{proof}
Let $\mathcal{F}(X)$ be the set of all non-empty closed subsets of $X$ endowed with the topology of closed convergence (which here coincides with the topology induced by the Hausdorff metric; see Chapters 3 and 4 in \cite{kleinthompson}). Since $X$ is compact, $\mathcal{F}(X)$ is compact and metrizable.

Consider the subgame triggered if Player 2 rejects $A_1$. Player 2 must choose a set $A_2 \in \mathcal{A}$, where $\mathcal{A} = \{ F \in \mathcal{F}(X) \mid \nu(F) \geq \nu(A_1) \}$. Since the measure $\nu$ is upper semicontinuous with respect to the Hausdorff metric (i.e., $\limsup \nu(A_n) \leq \nu(\lim A_n)$), the set $\mathcal{A}$ is closed, and therefore compact.  Now, since $A_2$ is a closed subset of a compact space, Player 1 has an optimal choice in any menu $A_2\in \mathcal{A}$ offered by Player 2.

Let $\beta_1(A_2)$ denote the set of optimal choices for Player 1 from a menu $A_2$. By the Maximum Theorem, the correspondence $A_2 \mapsto \beta_1(A_2)$ is upper hemicontinuous and compact-valued.
Define the supremum of Player 2's potential utility from a counter-offer as:
\[
    \hat{v} = \sup \{ \nu(L_2(y)) \mid A_2 \in \mathcal{A}, y \in \beta_1(A_2) \}.
\]
Since $\mathcal{A}$ is compact and the graph of $\beta_1$ is closed (by upper hemicontinuity), the set of reachable outcomes $Y = \bigcup_{A_2 \in \mathcal{A}} \beta_1(A_2)$ is compact.  By Lemma \ref{lem:cont}, the function  $y \mapsto \nu(L_2(y))$ is continuous; hence, the supremum $\hat{v}$ is attained at some outcome $\hat{x} \in Y$. There exists therefore a specific counter-offer $\hat{A}_2 \in \mathcal{A}$ such that $\hat{x} \in \beta_1(\hat{A}_2)$.

We specify Player 1's strategy $\sigma_1$ to choose $\hat{x}$ if Player 2 proposes $\hat{A}_2$. For any other menu $A'\in\mathcal{A}$ that Player 2 may choose, Player 1 chooses an arbitrary element of $\beta_1(A')$. Given this strategy, Player 2's best response is defined as follows:
\begin{enumerate}
    \item Identify the optimal choice $x \in A_1$ (maximizing $\nu(L_2(x))$; which exists by Lemma~\ref{lem:cont}).
    \item Compare $\nu(L_2(x))$ with $\hat{v} = \nu(L_2(\hat{x}))$.
    \item If $\nu(L_2(x)) \ge \hat{v}$, accept $A_1$ and choose $x$. Otherwise, reject and propose $\hat{A}_2$.
\end{enumerate}
This constitutes a best response for Player 2.
\end{proof}

\begin{lemma}\label{lemma:ex} For any compromise solution $x^*$, there is a subgame-perfect Nash equilibrium in which the outcome is $x^*$.
\end{lemma}
 \begin{proof}Consider the strategy for Player 1 defined in the proof of Lemma~\ref{lem:2opt} in which: 
 \begin{enumerate}
     \item she proposes the set $L_2(x^*)$ in stage 1, and
     \item she chooses an optimal outcome from any proposal $A_2$ by Player 2.
 \end{enumerate} 
 Fix the strategy for Player 2 that was defined in Lemma~\ref{lem:2opt}. In particular, we shall see that in the present case (for the particular choice of $A_1=L_2(x^*)$) this strategy is:
 \begin{enumerate}
 \item In the subgame with $A_1=L_2(x^*)$, accept $x^*$.
\item In any subgame in which  $A_1\neq L_2(x^*)$,  Player 2 follows the optimal strategy described in Lemma~\ref{lem:2opt}.
 \end{enumerate} 

We now verify that this pair of strategies constitutes a subgame perfect Nash equilibrium.

If Player 2 does not reject Player 1's on-path proposal $L_2(x^*)$, it is optimal for her to choose $x^*$ from $L_2(x^*)$. If Player 2 were to reject the proposal, then she needs to propose $A_2$, closed, with $\nu(A_2)\geq\nu (L_2(x^*))$. Say that Player 1 chooses $x'\in A_2$ in response to $A_2$. This means that $L_1(x')\supseteq A_2$. Then Lemma~\ref{lem:PO} implies that: $$\nu (L_1(x'))\geq \nu(A_2)\geq  \nu (L_2(x^*))=\nu(L_1(x^*)).$$ By Pareto optimality of $x^*$ this means that $\nu(L_2(x'))\leq \nu (L_2(x^*))$. So the rejection of $L_2(x^*)$ by 2 is not an improvement for 2. We conclude that it is optimal for Player 2 to accept the proposal and choose $x^*$, as stipulated by the strategies we have proposed. That the strategy is sequentially rational for Player 2 was shown in Lemma~\ref{lem:2opt}.

We proceed to verify that Player 1's proposal of $L_2(x^*)$ is optimal.
  
Consider an alternative proposal for Player 1, $A_1\neq L_2(x^*)$, and suppose that Player 2 chooses $x'\in A_1$. For this to be a profitable deviation, we must have that $x'\succ_1 x^*$, which implies that $\nu (L_1(x'))>\nu(L_1(x^*))$. By Pareto optimality of $x^*$, $\nu (L_2(x'))<\nu(L_2(x^*))$.

Since $A_1\subseteq L_2(x')$, Lemma~\ref{lem:PO} implies:
\[
\nu(A_1)\leq \nu(L_2(x'))< \nu(L_2(x^*))=\nu(L_1(x^*)).
\]
Hence there exists $\ep>0$ such that $\nu(L^\ep_1(x^*))\geq \nu(A_1)$. Player 2 is therefore better off rejecting $A_1$ and instead proposing $L^\ep_1(x^*)$, from which player 1 finds it (uniquely) optimal to choose $x^*$. Consequently, proposing $A_1$ is not strictly preferable for player~1.

Finally, consider an alternative proposal for Player 1, $A_1\neq L_2(x^*)$, and suppose that Player 2 rejects $A_1$. Following the rejection, 2 proposes $A_2$, from which Player 1 chooses $x'$. For this to be a profitable deviation by 1 we must have that $x'\succ_1 x^*$. 

By the Pareto optimality of $x^*$, we have $x^*\succ_2 x'$. This means that $\nu(L_2(x^*))>\nu(L_2(x'))$ by Lemma \ref{lem:largeset} . Player 2 could have ensured $x^*$ by proposing $L^\ep_1(x^*)$ for some $\ep>0$. So the set $L^\ep_1(x^*)$ must not have been feasible, meaning that $\nu(L^\ep_1(x^*))<\nu(A_1)$ for all $\ep>0$. Thus $\nu(L_2(x^*))=\nu(L_1(x^*))\leq\nu(A_1)$.  The rejection by 2 of $A_1$ means that $L_2(x')\supseteq A_1$. So we obtain the following inequalities:
$$\nu(L_2(x'))\geq \nu(A_1) \geq \nu(L_1(x^*))= \nu(L_2(x^*)) >\nu(L_2(x')),$$ which is absurd. We conclude that there is no such profitable deviation for Player~1.  
\end{proof}

\begin{lemma}\label{lemma:outcome} If $x$ is the outcome of a subgame-perfect Nash equilibrium, then $x$ is a compromise solution. 
\end{lemma}

\begin{proof} Consider a subgame-perfect Nash equilibrium with outcome $x$. Suppose, towards a contradiction, that $x$ is not a compromise solution and let $x^*$ be a compromise. Then $\nu(L_i(x))<\nu(L_i(x^*))$ for either 1 or 2 (or for both). We consider each case in turn.

Suppose first that the inequality holds for $i=1$. 

Then $\nu(L_1(x))<\nu(L_1(x^*))$ implies that $x^*\succ_1 x$. By continuity of $\succeq_1$, there is a neighborhood $N_{x^*}$ of $x^*$ such that $\tilde x\succ_1 x$ for all $\tilde x\in N_{x^*}$. By local non-satiation of $\succeq_2$ at $x^*$, there is $x'\in N_{x^*}$ with $x'\succ_2 x^*$. Let $A'_1=L_2(x^*)\cup\{x'\}$, which is closed and has measure $\nu(A'_1)=\nu(L_2(x^*))$, since $\nu$ is atomless. In $A'_1$, $x'$ is the unique optimal outcome for Player 2. Player 1 may deviate and offer $A'_1$ instead of the on-path offer of $A_1$. 

If Player 2 accepts $A'_1$, then she would choose $x'$. The outcome of $x'$ would make $A'_1$ a profitable deviation for Player 1. A contradiction.

If Player 2 rejects $A'_1$, then she must offer Player 1 a closed set $A'_2$ with $\nu(A'_2)\geq \nu (A'_1)= \nu(L_2(x^*))$. Let $x''$ be 1's choice from $A'_2$ according to the equilibrium strategies: $A'_2\subseteq L_1(x'')$. Then we must have: $$\nu(L_1(x''))\geq \nu(A'_2)\geq \nu(L_2(x^*)) = \nu(L_1(x^*))\implies  x''\succeq_1 x^*.$$ This, however, means that $x^*\succeq_2 x''$, as $x^*$ is Pareto optimal. Thus, 2 would have been better off accepting $A'_1$ and choosing $x'\succ_2 x^*$. Again we conclude that Player 1 would be better off offering $A'_1$ and obtain a contradiction with the definition of equilibrium.

Suppose now that the inequality does not hold for $i=1$. That is, $\nu(L_1(x)) \ge \nu(L_1(x^*))$.
By Pareto optimality (or simply the fact that $x$ is not a compromise while $x^*$ is), it must be that $\nu(L_2(x)) < \nu(L_2(x^*))$, which implies $x^* \succ_2 x$.

In the putative equilibrium, let $A_1$ be the proposal made by Player 1.
\begin{itemize}
    \item If Player 2 accepts $A_1$ and chooses $x$, then $x$ is the $\succeq_2$-maximal element in $A_1$, so $A_1 \subseteq L_2(x)$.
    \item If Player 2 rejects $A_1$, makes a counter-proposal, and Player 1 chooses $x$, then it must be that for all $y \in A_1$, $x \succeq_2 y$ (otherwise Player 2 would have accepted $y$). Thus, $A_1 \subseteq L_2(x)$.
\end{itemize}
In either case, $\nu(A_1) \le \nu(L_2(x))$.

Let $\varepsilon > 0$ be such that $\nu(L_2(x)) < \nu(L_2(x^*)) - \varepsilon$.
Consider a deviation where Player 2 rejects $A_1$ and offers the closed set $L^\varepsilon_1(x^*)$. We check feasibility:
\[
    \nu(A_1) \le \nu(L_2(x)) < \nu(L_2(x^*)) - \varepsilon = \nu(L_1(x^*)) - \varepsilon \le \nu(L^\varepsilon_1(x^*)).
\]
(The last inequality follows from Lemma \ref{lem: ubs}).
Since $\nu(L^\varepsilon_1(x^*)) > \nu(A_1)$, this is a legitimate counter-proposal.
From $L^\varepsilon_1(x^*)$, the unique best choice for Player 1 is $x^*$.
Since $x^* \succ_2 x$, this rejection followed by the counter-proposal $L^\varepsilon_1(x^*)$ yields a strictly better outcome for Player 2 than the equilibrium outcome $x$.
This is a profitable deviation, contradicting the assumption that $x$ is an equilibrium outcome.

\end{proof}

\subsection{Proof of Proposition~\ref{prop:POmu}}
Because $\nu$ has full support on $X$, any Borel set $A$ with a non-empty interior satisfies $\nu(A) > 0$. Let $U^O_i(x')$ denote the strict upper contour set of $\succeq_i$ at $x'\in X$. By continuity and local non-satiation, $U^O_i(x)$ is a non-empty open set.  

Since $x$ is Pareto optimal, for any $y \in U^O_2(x)$ we must have $x \succ_1 y$. This implies $U^O_2(x) \subseteq L_1(x)$. Symmetrically, $U_1(x) \subseteq L_2(x)$. Let $\Sigma$ be the finite algebra of subsets of $X$ generated by $L_1(x)$ and $L_2(x)$. The four atoms of $\Sigma$ are:
\begin{itemize}
    \item $A_{10} = L_1(x) \setminus L_2(x) = L_1(x) \cap U^O_2(x) = U^O_2(x)$. Thus, $\nu(A_{10}) > 0$.
    \item $A_{01} = L_2(x) \setminus L_1(x) = L_2(x) \cap U_1(x) = U_1(x)$. Thus, $\nu(A_{01}) > 0$.
    \item $A_{00} = X \setminus (L_1(x) \cup L_2(x)) = U_1(x) \cap U^O_2(x) = \emptyset$. Thus, $\nu(A_{00}) = 0$.
    \item $A_{11} = L_1(x) \cap L_2(x)$. 
\end{itemize}

We may define a probability measure $\mu_0$ on $(X, \Sigma)$ such that $\mu_0(L_1(x)) = \mu_0(L_2(x))$, which is equivalent to $\mu_0(A_{10}) = \mu_0(A_{01})$. In fact, we may ensure that for every atom $A_k$, if $\nu(A_k) > 0$ then $\mu_0(A_k) > 0$. To ensure absolute continuity, we ensure that $\nu(A_k) = 0 \implies \mu_0(A_k) = 0$: Specifically,  set $\mu_0(A_{00}) = 0$, $\mu_0(A_{10}) = \mu_0(A_{01}) =\alpha$ and $\mu_0(A_{11})=1-2\alpha$ with $1/2\geq \alpha>0$, and $\alpha=1/2$ only if $\nu(A_{11})=0$. 

Since $\mu_0$ is absolutely continuous with respect to the restriction of $\nu$ to $\Sigma$, we may let $f = \frac{d\mu_0}{d\nu}$ be the corresponding Radon-Nikodym derivative; $f$ is a strictly positive step function defined by the constants $c_k = \frac{\mu_0(A_k)}{\nu(A_k)}$ on each atom $A_k$ where $\nu(A_k) > 0$. 

We extend $\mu_0$ to a measure $\mu$ on the full Borel $\sigma$-algebra of $X$ by 
\[ \mu(B) = \int_B f \diff \nu \quad \text{for all Borel sets } B \subseteq X. \]
Because $\nu$ is a regular Borel probability measure with full support, and $f$ is strictly positive $\nu$-almost everywhere, the extended measure $\mu$ is a Borel probability measure that inherits the full support of $\nu$. Furthermore, by construction, $\mu(L_1(x)) = \mu_0(L_1(x)) = \mu_0(L_2(x)) = \mu(L_2(x))$. Let this common value be denoted as $M$.

To show that $x$ is a compromise, we argue by contradiction. Suppose that there exists some alternative  $x' \in X$ with 
\[ \min \{ \mu(L_1(x')), \mu(L_2(x')) \} > \min \{ \mu(L_1(x)), \mu(L_2(x)) \} = M \]

This means that $\mu(L_i(x')) > \mu(L_i(x))$ for $i=1,2$. Now, because $\tilde x\mapsto \mu(L_i(\tilde x))$ is a utility representation of $\succeq_i$, we conclude that $x'$ Pareto dominates $x$, which is absurd.

\subsection{Proof of Proposition~\ref{prop:PE}}
We wish to find the Lebesgue measure $\nu(L(x_0,g_0))$ of the set:
\[ L(x_0,g_0) = \{ (x,g) \in \Re^2 : x \geq 0, g > 0, x + \theta \log(g) \leq U \},
\] where $U=x_0+\theta\log(g_0)$. The condition $x \leq U - \theta \log(g)$ and $x \geq 0$ imply that $U - \theta \log(g) \geq 0$. So $g$ has to satisfy $0 < g \leq e^{U/\theta}$. Using Fubini's Theorem, the Lebesgue measure is given by the iterated integral:
\[
\nu ((L(x_0,g_0)) = \int_{0}^{e^{U/\theta}} \left( \int_{0}^{U - \theta \log(g)} 1 \diff x \right) \diff g
 = \int_{0}^{e^{U/\theta}} (U - \theta \log(g)) \, dg.
\]

Now, this gives us :

\begin{align*}
\int_{0}^{e^{U/\theta}} (U - \theta \log(g)) \, dg  &= \left[ Ug - \theta(g \log(g) - g) \right]_{0}^{e^{U/\theta}} \\
&= \left[ g(U + \theta - \theta \log(g)) \right]_{0}^{e^{U/\theta}} \\
&= e^{U/\theta} \left( U + \theta - \theta \frac{U}{\theta} \right) - 0 = \theta e^{U/\theta}.
\end{align*}

At a compromise solution, we have $g^*=\theta_1+\theta_2,$ and  $\nu(L_1(x^*_1+\theta_1\log (g^*)))=\nu(L_2(x^*_1+\theta_1\log (g^*)))$. This means that:
\[
  \theta_1 e^{\frac{x^*_1+\theta_1\log (g^*)}{\theta_1}}
  =\theta_2 e^{\frac{x^*_2+\theta_2\log (g^*)}{\theta_2}}.
\]

Hence (canceling $\log(g^*)$), 
\[
\log(\theta_1) + \frac{x^*_1}{\theta_1} = \log(\theta_2) + \frac{x^*_2}{\theta_2} 
\]
Since $x^*_1+x^*_2=1-\theta_1-\theta_2,$
\[\log\left(\frac{\theta_1}{\theta_2}\right) = \frac{1 - \theta_1 - \theta_2 - x^*_1}{\theta_2} - \frac{x^*_1}{\theta_1} 
= \frac{1 - \theta_1 - \theta_2}{\theta_2} - x^*_1\left(\frac{\theta_1 + \theta_2}{\theta_1\theta_2}\right).\]

Solving for $x^*_1$:
\begin{align*}
x^*_1\left(\frac{\theta_1 + \theta_2}{\theta_1\theta_2}\right) &= \frac{1 - \theta_1 - \theta_2}{\theta_2} - \log\left(\frac{\theta_1}{\theta_2}\right).
\end{align*}

This leads to:
\begin{align*}
x^*_1 &= \frac{\theta_1(1 - \theta_1 - \theta_2)}{\theta_1 + \theta_2} - \frac{\theta_1\theta_2}{\theta_1 + \theta_2}\log\left(\frac{\theta_1}{\theta_2}\right)\\
x^*_2 &= (1 - \theta_1 - \theta_2) - x^*_1 = \frac{\theta_2(1 - \theta_1 - \theta_2)}{\theta_1 + \theta_2} + \frac{\theta_1\theta_2}{\theta_1 + \theta_2}\log\left(\frac{\theta_1}{\theta_2}\right),
\end{align*}
as desired. \qed

\subsection{Proof of Theorem~\ref{theo:fs}.}

First, we shall determine the measure of the lower contour set for any allocation $(x, 1-x)$ for an agent with F-S preferences.

\begin{proposition}\label{prop:fs} Assume that Players 1 and 2 have Fehr-Schmidt preferences over  $X = \{ (z_1, z_2) \in \mathbb{R}^2_+ \mid z_1 + z_2 \le 1 \}$ with $\beta_1,\beta_2<1/2$. For any allocation $(x,1-x)$, the Lebesgue measure of its lower contour for Player 1 set equals:
\begin{equation*}
\nu(L_1(x,1-x)) = \begin{cases}
    \frac{1+2\alpha_1}{2\alpha_1} x^2, & \text{if } 0 \le x < \frac{\alpha_1}{1+2\alpha_1}, \\[10pt]
    \frac{((1+2\alpha_1)x - \alpha_1)^2}{2(1-\beta_1)} + 0.25 - \frac{1}{4}(1+2\alpha_1)(1-2x)^2, & \text{if } \frac{\alpha_1}{1+2\alpha_1} \le x \le 1/2, \\[10pt]
    \frac{1}{2} - \frac{1-2\beta_1}{2(1-\beta_1)}(1-x)^2, & \text{if } x > 1/2,
\end{cases}
\end{equation*}
where $\nu(L_1(1/2,1/2)) = \frac{3-2\beta_1}{8(1-\beta_1)}$ and the one of Player 2 is:

\begin{equation*} 
\nu(L_2(x,1-x)) = \begin{cases}  \frac{1}{2} - \frac{1-2\beta_2}{2(1-\beta_2)}x^2, & \text{if } x < 1/2, \\[10pt]  \frac{((1+2\alpha_2)x - (1+\alpha_2))^2}{2(1-\beta_2)} + 0.25 - \frac{1}{4}(1+2\alpha_2)(1-2x)^2, & \text{if } 1/2 \le x \le \frac{1+\alpha_2}{1+2\alpha_2}, \\[10pt]  \frac{1+2\alpha_2}{2\alpha_2} (1-x)^2, & \text{if } x > \frac{1+\alpha_2}{1+2\alpha_2}. \end{cases} 
\end{equation*}
and
$\nu(L_2(1/2,1/2)) = \frac{3-2\beta_2}{8(1-\beta_2)}$.
\end{proposition}

In an abuse of notation, in this proof we shall often write $L_i(x)$ for $\nu(L_i(x,1-x))$ and $U_i(x)$ for $\nu(X)\setminus \nu(L_i(x,1-x))$. 

\begin{proof}
Assume that $\beta<1/2$. This means that along the Pareto frontier, the marginal utility w.r.t.\ $x$ is positive for any value of $x$ (i.e., $(1+2\alpha)$ for $x<1/2$ and $(1-2\beta)$ otherwise).

We will derive it for Player 1 and obtain by symmetry for Player 2. To this end, we divide the outcome space $X$ in two regions: the one in which Player 1 gets more, where 
$z_1\geq z_2$, and hence he feels guilt with respect to Player 2 and those in which Player 1 gets less, where 
$z_1\leq z_2$, and hence he envies Player 2 as argued by F-S. We refer to the line $z_1+z_2=1$ as the budget line.

Allocation $(x,1-x)$ with $x < 0.5$: Envy

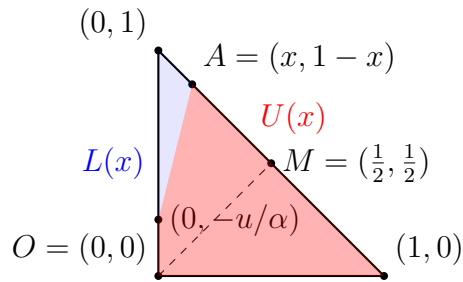
\begin{figure}[ht]
        \centering
        \begin{tikzpicture}[scale=3]
            \draw[thick] (0,0) -- (1,0) -- (0,1) -- cycle;
            \draw[dashed] (0,0) -- (0.5,0.5);
                        \coordinate (A) at (0.15, 0.85);
            \coordinate (D) at (0, 0.25);
            \coordinate (M) at (0.5, 0.5);
                        \coordinate (X) at (1, 0);
            \coordinate (Y) at (0, 1);
            \coordinate (O) at (0, 0);
                        \fill (A) circle (0.5pt) node[above right] {$A=(x,1-x)$};
            \fill (M) circle (0.5pt) node[ right] {$M=(\frac{1}{2},\frac{1}{2})$};   
            \fill (D) circle (0.5pt) node[right] {$(0,-u/\alpha)$};
                      \fill (O) circle (0.5pt) node[above left] {$O=(0,0)$};
                 \fill (X) circle (0.5pt) node[above right] {$(1,0)$};
            \fill (Y) circle (0.5pt) node[above left] {$(0,1)$};
                       \node[red, scale=1] at (0.6, 0.7) {\text{$U(x)$}};
            \node[blue, scale=1] at (-0.2, 0.5) {\text{$L(x)$}};
                                     \fill[blue, opacity=0.1] (O) -- (D) -- (A) -- (Y) -- cycle;
                          \fill[red, opacity=0.3] (A) -- (D) -- (O) -- (X) -- cycle;

        \end{tikzpicture}
         \caption{ Negative utility and an allocation inducing envy: $x < 0.5$ }
         \label{fig:envy0}
\end{figure}

\begin{figure}[hht]
        \centering
        \begin{tikzpicture}[scale=3]
            \draw[thick] (0,0) -- (1,0) -- (0,1) -- cycle;
            \draw[dashed] (0,0) -- (0.5,0.5);
                        \coordinate (A) at (0.35, 0.65);
            \coordinate (D) at (0.25, 0.25);
            \coordinate (M) at (0.5, 0.5);
            \coordinate (C) at (0.85, 0);
            \coordinate (X) at (1, 0);
            \coordinate (Y) at (0, 1);
            \coordinate (O) at (0, 0);
                        \fill (A) circle (0.5pt) node[above right] {$A=(x,1-x)$};
            \fill (M) circle (0.5pt) node[ right] {$M=(\frac{1}{2},\frac{1}{2})$};   
            \fill (D) circle (0.5pt) node[right] {$D=(u,u)$};
            \fill (C) circle (0.5pt) node[below ] {$C=(\frac{u}{1-\beta},0)$};
            \fill (O) circle (0.5pt) node[above left] {$O=(0,0)$};
                 \fill (X) circle (0.5pt) node[above right] {$(1,0)$};
            \fill (Y) circle (0.5pt) node[above left] {$(0,1)$};
                       \node[red, scale=1] at (0.6, 1) {\text{$U(x)$}};
            \node[blue, scale=1] at (-0.2, 0.5) {\text{$L(x)$}};
                          \fill[blue, opacity=0.1] (O) -- (D) -- (C) -- cycle;
             \fill[blue, opacity=0.1] (O) -- (D) -- (A) -- (Y) -- cycle;
             \fill[red, opacity=0.3] (D) -- (M) -- (X) -- (C) -- cycle;
             \fill[red, opacity=0.3] (A) -- (M) -- (D) -- cycle;

        \end{tikzpicture}
         \caption{Positive utility and an allocation inducing envy: $x < 0.5$}
         \label{fig:envy}
\end{figure}
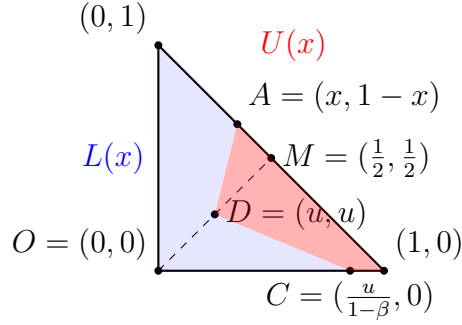

Remark that for any $(x,1-x)$ with $x<1/2$, we obtain that:
\begin{equation}\label{eq:uexplicit2}
    u = U(x,1-x) \Longleftrightarrow u = (1+2\alpha)x - \alpha.
\end{equation}

The allocation $(x,1-x)$ that gives a utility of 0 equals $(1+2\alpha)x - \alpha=0$ so $x=\frac{\alpha}{1+2\alpha}$.

When $x<\frac{\alpha}{1+2\alpha}$, the $LCS$ is simple as shown by Figure \ref{fig:envy0} as it coincides with the area of a triangle with vertices $(0,1)$, $(x,1-x)$ and $(0,-u/\alpha)$. Its area is then $\frac{1+2\alpha}{2\alpha} x^2$ as the base is $1-(-u/\alpha)$, the height is $x$, and $\alpha+u=(1+2\alpha)x$.

When $x\geq \frac{\alpha}{1+2\alpha}$, the $LCS$ spans both sides of the line $z_1=z_2$ as shown by Figure \ref{fig:envy}. The $LCS$ of $(x,1-x)$ in the Guilt region, in which $z_1 \ge z_2$, is the triangle $OCD$ defined by the origin $O$, $D=(u,u)$, and $C=(\frac{u}{1-\beta}, 0)$. Its Lebesgue measure equals:
    \[ \nu(\LCS(x,1-x)\cap\{(z_1,z_2)\in X:z_1\geq z_2\}) = \frac{u^2}{2(1-\beta)}= \frac{((1+2\alpha)x - \alpha)^2}{2(1-\beta)}.  \]
    
The $UCS$ of $(x,1-x)$ in the Envy region, in which $z_1 < z_2$, is the triangle  $AMD$ formed by $A=(x, 1-x)$, $M=(1/2, 1/2)$, and  $D=(u,u)$. Its base equals $\sqrt{2}(1+2\alpha)(0.5-x)$, as the distance from $M$ to $D$ equals $\sqrt{(0.5-u)^2 + (0.5-u)^2} = \sqrt{2}(0.5-u)$ and $0.5 - u = 0.5 - ((1+2\alpha)x - \alpha) = (1+2\alpha)(0.5-x)$. Its height is $\frac{1-2x}{\sqrt{2}}$, that is,  the distance from $A$ to $M$. Therefore, 
        \[ \nu(\UCS(x,1-x)\cap\{(z_1,z_2)\in X:z_1\leq z_2\}) = \frac{1}{2} \cdot \left[ \sqrt{2}(1+2\alpha)(0.5-x) \right] \cdot \left[ \frac{1-2x}{\sqrt{2}}. \right] \]
        Since $(0.5-x) = \frac{1}{2}(1-2x)$:
        \[ \nu(\UCS(x,1-x)\cap\{(z_1,z_2)\in X:z_1\leq z_2\}) = \frac{1}{4}(1+2\alpha)(1-2x)^2 .\]

    The  $LCS$ area in this region is $\frac{1}{4} - \nu(\UCS(x,1-x)\cap\{(z_1,z_2)\in X:z_1\leq z_2\})$.
    
    The Lebesgue measure of the $LCS$ of $(x,1-x)$ for $\frac{\alpha}{1+2\alpha}<x < \frac{1}{2}$ is:
    \[ L_1(x) = \frac{((1+2\alpha)x - \alpha)^2}{2(1-\beta)} + \frac{1}{4} - \frac{1}{4}(1+2\alpha)(1-2x)^2. \]

Allocation $(x,1-x)$ with $x \ge 0.5$: Guilt

Any indifference curve for Player 1 in the Guilt region is decreasing in $x$ as the equation is:
\[
    y = \frac{u}{\beta} - \frac{1-\beta}{\beta} x.
\]
The magnitude of the slope is $\frac{1-\beta}{\beta}$. Observe that:
\[
    \frac{1-\beta}{\beta} > 1 \iff 1-\beta > \beta \iff 1 > 2\beta \iff 1/2 > \beta.
\]
Thus, the indifference curve is steeper than the budget line (which has a slope of -1).

Moreover, along the Pareto frontier, the marginal utility with respect to $x$ is positive for any value of $x$. Therefore, the upper contour set at $(x, 1-x)$ lies to the right of the indifference curve.

Remark that for any $u \in [1/2, 1-\beta]$ and $x \ge 1/2$:
\begin{equation}\label{eq:uexplicit}
    u = U(x, 1-x) \iff u = (1-2\beta)x + \beta \iff u = (1-\beta)x + \beta(1-x).
\end{equation}

The upper contour set at $(x, 1-x)$ is the triangle bounded by the indifference line of level $u$, the $x$-axis, and the budget line. Specifically, $\UCS(x)$ coincides with the triangle $ABC$ shown in Figure \ref{fig:guilt0}, with vertices: $A=(x, 1-x)$, $B=(1,0)$, and $C=(z_C, 0)$. Solving $(1-\beta)z_C = u$ for the x-intercept gives $z_C = \frac{u}{1-\beta}$. Substituting $u$ from (\ref{eq:uexplicit}), we obtain:
\[
    z_C = \frac{(1-\beta)x + \beta(1-x)}{1-\beta}.
\]

\begin{figure}[hht]
    \centering
    \begin{tikzpicture}[scale=3]
        \draw[thick] (0,0) -- (1,0) -- (0,1) -- cycle;
        \draw[dashed] (0,0) -- (0.5,0.5);
        \coordinate (A) at (0.7, 0.3);
        \coordinate (D) at (1, 0);
        \coordinate (Y) at (0, 1);
        \coordinate (C) at (0.85, 0);
        \coordinate (O) at (0, 0);

        \fill (Y) circle (0.5pt) node[above left] {$(0,1)$};
        \fill (A) circle (0.5pt) node[above right] {$A=(x,1-x)$};
        \fill (D) circle (0.5pt) node[below right] {$B=(1,0)$};
        \fill (C) circle (0.5pt) node[below left] {$C=(\frac{u}{1-\beta},0)$};
        \fill (O) circle (0.5pt) node[below left] {$(0,0)$};

        \fill[red, opacity=0.3] (A) -- (1,0) -- (C) -- cycle;
        \node[red, scale=1] at (1.1, 0.15) {$U(x)$};
        \fill[blue, opacity=0.1] (0,0) -- (0,1) -- (A) -- (C) -- cycle;
        \node[blue, scale=1] at (0.3, 0.3) {$L(x)$};
    \end{tikzpicture}
    \caption{An allocation inducing guilt: $x \ge 0.5$. The UCS is the red triangle $ABC$.}
    \label{fig:guilt0}
\end{figure}
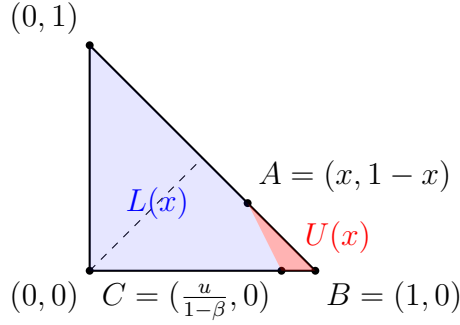

The base of the upper contour set triangle (segment $CB$) has length:
\[
    1 - z_C = 1 - \frac{(1-\beta)x + \beta(1-x)}{1-\beta} = (1-x)\frac{1-2\beta}{1-\beta}.
\]
The height corresponds to the $y$-coordinate of $A$, which is $h = 1-x$. The area is therefore:
\[
    \nu(\UCS) = \frac{1}{2} \cdot \text{Base} \cdot h = \frac{1}{2} \left[ \frac{1-2\beta}{1-\beta}(1-x) \right] (1-x) = \frac{1-2\beta}{2(1-\beta)}(1-x)^2.
\]
Consequently, the Lebesgue measure of the lower contour set in this region is the total area minus the area of the upper contour set:
\[
    L_1(x) = \nu(X) - \nu(\UCS) = \frac{1}{2} - \frac{1-2\beta}{2(1-\beta)}(1-x)^2.
\]
The result for Player 2 follows by symmetry.
\end{proof}

\begin{proposition}\label{prop:monL}
The measure of the lower contour set $L_1(x)$ is strictly monotonically increasing on $[0,1]$ assuming $\beta_1 < 1/2$. By symmetry, $L_2(x)$ is strictly monotonically decreasing.
\end{proposition}

\begin{proof}
Let $x, x' \in [0,1]$ with $x' > x$.
First, recall that when $\beta_1 < 1/2$, the marginal utility of Player 1 along the feasibility frontier is strictly positive. Thus, the utility levels satisfy $u' := U_1(x', 1-x') > U_1(x, 1-x) := u$.

Abusing notation, the lower contour set for a given utility level $v$ is $\LCS(v) = \{ z \in X \mid U_1(z) \le v \}$.
Since preferences are continuous and non-satiated, a strictly higher utility level implies a strictly larger lower contour set. Specifically, for $u' > u$:
\[
    \LCS(u) \subset \text{int}(\LCS(u')),
\]
which implies that the Lebesgue measures satisfy $\nu(L_1(x')) > \nu(L_1(x))$ as these are full dimensional.
Thus, $\nu(L_1(x))$ is strictly increasing in $x$. The result for Player 2 follows by symmetry.
\end{proof}

\paragraph*{Completion of the Proof of Theorem \ref{theo:fs}}

We have established that $L_1(x)$ is strictly increasing and $L_2(x)$ is strictly decreasing on $[0,1]$. As these are continuous functions, a unique intersection $x^*$ exists where $L_1(x^*) = L_2(x^*)$.

To determine the location of $x^*$ relative to the midpoint, we evaluate the measures at $x = 1/2$. Recall from  Proposition \ref{prop:fs} that:
\[
    L_i(1/2) = \frac{3-2\beta_i}{8(1-\beta_i)}.
\]
Let $f(\beta) = \frac{3-2\beta}{8(1-\beta)}$. We analyze the sensitivity of this value to the guilt parameter by taking the derivative with respect to $\beta$:
\[
    f'(\beta) = \frac{1}{8} \cdot \frac{-2(1-\beta) - (-1)(3-2\beta)}{(1-\beta)^2} = \frac{1}{8} \cdot \frac{-2 + 2\beta + 3 - 2\beta}{(1-\beta)^2} = \frac{1}{8(1-\beta)^2}.
\]
Since $\beta < 1$, $f'(\beta) > 0$. This implies that a player with higher guilt aversion has a strictly larger lower contour set at the equal split $x=1/2$.

We now consider the three cases from the theorem:
\begin{enumerate}
    \item \textbf{Case $\beta_1 = \beta_2$:} Then $L_1(1/2) = L_2(1/2)$. Since $L_1$ is increasing and $L_2$ is decreasing, the unique intersection must occur exactly at $x^* = 1/2$.
    \item \textbf{Case $\beta_1 > \beta_2$:} Since $f(\beta)$ is increasing, $L_1(1/2) > L_2(1/2)$. Given the monotonicity of the functions, the intersection $L_1(x^*) = L_2(x^*)$ must occur at some $x^* < 1/2$.
    \item \textbf{Case $\beta_1 < \beta_2$:} Similarly, $L_1(1/2) < L_2(1/2)$, implying the intersection must occur at some $x^* > 1/2$.
\end{enumerate}

\bibliographystyle{ecta}
\bibliography{consensus} 

\clearpage
\appendix
\section{Examples illustrating Theorem~\ref{thm:commonnu}}\label{sec:appexamples}

In this appendix, we develop some examples to further illustrate the role of the condition in Theorem~\ref{thm:commonnu}, and what the effect of a common cardinalization is on the utility possibility set. 

\subsection{One dimensional example}

Let the space of alternatives be $X = [0,1]$ and consider the following two utility functions for the agents. 
\begin{align*}
    u_1(x) = 1 - x^2 \: \text{ and } \:     u_2(x) = 1 - 4\left(x - \frac{1}{2}\right)^2.
\end{align*}
Both functions are surjective onto $[0,1]$. Because $u_1$ is strictly decreasing on $[0,1]$, we can express $x$ directly as a function of utils for 1, $v_1$: $x = \sqrt{1 - v_1}$.
Substituting this into $u_2(x)$ yields the utility possibility set:
\[
  \UPS = \{ (v_1, v_2) \in [0,1]^2 : v_2 = 4v_1 + 4\sqrt{1-v_1} - 4 \}.\]
So the UPS is a  curve that starts at $(1,0)$, has a maximum at $(0.75, 1)$, and ends at $(0,0)$. See the illustration on the left in Figure~\ref{fig:onedimex}.

\begin{figure}[t]
    \centering
    \begin{tikzpicture}
        \begin{axis}[
            width=0.45\textwidth,
            height=0.45\textwidth,
            title={\textbf{$\UPS$}},
            xlabel={$v_1$},
            ylabel={$v_2$},
            xmin=0, xmax=1.1,
            ymin=0, ymax=1.2,
            axis lines=left,
                       grid style={line width=.1pt, draw=gray!30},
            major grid style={line width=.2pt,draw=gray!50},
            xtick={0, 0.25, 0.5, 0.75, 1},
            ytick={0, 0.25, 0.5, 0.75, 1},
                                    extra x tick style={tick label style={font=\tiny, color=green!50!black, yshift=1ex}},
                                   extra y tick style={tick label style={font=\tiny, color=red}},
                   ]
        
                \fill[red, opacity=0.1] (0, 0) rectangle (1, 0.36);
        \fill[green!50!black, opacity=0.15] (0, 0) rectangle (0.19, 1);
        \fill[green!50!black, opacity=0.15] (0.99, 0) rectangle (1, 1);
        
        \addplot [
            domain=0:1, 
            samples=200, 
            color=blue,
            very thick,
        ]
        {4*x + 4*sqrt(1-x) - 4};
        
        \node[red, anchor=west, font = \tiny] at (axis cs: 0.3, 0.4) {$E=[0,0.36]$};
        \node[green!50!black, anchor=west, font = \tiny] at (axis cs: 0.1, 1.1) {$\UPS(E)=[0,0.19]\cup [0.99,1]$};
        
        \end{axis}
    \end{tikzpicture}
    \hfill
    \begin{tikzpicture}
        \begin{axis}[
            width=0.45\textwidth,
            height=0.45\textwidth,
            title={\textbf{$\hat{\UPS}$}},
            xlabel={$v_1$},
            ylabel={$v_2$},
            xmin=0, xmax=1.1,
            ymin=0, ymax=1.2,
            axis lines=left,
                       grid style={line width=.1pt, draw=gray!30},
            major grid style={line width=.2pt,draw=gray!50},
            xtick={0, 0.25, 0.5, 0.75, 1},
            ytick={0, 0.25, 0.5, 0.75, 1},
                                extra x tick style={tick label style={font=\tiny, color=green!50!black, yshift=1ex}},
                                  extra y tick style={tick label style={font=\tiny, color=red}},
                    ]
        
                \fill[red, opacity=0.1] (0, 0) rectangle (1, 0.36);
        \fill[green!50!black, opacity=0.15] (0, 0) rectangle (0.18, 1);
        \fill[green!50!black, opacity=0.15] (0.82, 0) rectangle (1, 1);
        
        \addplot [
            domain=0:1, 
            samples=200, 
            color=black,
            very thick,
        ]
        {1 - abs(2*x - 1)};

        \node[red, anchor=west, font = \tiny] at (axis cs: 0.3, 0.4) {$E= [0,0.36]$};
        \node[green!50!black, anchor=west, font = \tiny] at (axis cs: 0.1, 1.1) {$\hat{\UPS}(E)=[0,0.18]\cup [0.82,1]$};
                \addplot[mark=*, black] coordinates {(2/3,2/3)};
        \node[anchor=south, font=\tiny] at (axis cs: .8, .66) {$(2/3,2/3)$};
        \end{axis}

      \end{tikzpicture}
    \caption{Utility possibility sets before (on the left) and after common cardinalization (on the right). The canonical utility associated to the compromise $x^*=1/3$ is denoted $\hat{u}(x^*)=(2/3,2/3)$.}
\label{fig:onedimex}
\end{figure}

To test the condition in the theorem, let $E = [0, 0.36]$ be a set where Agent 2 obtains their worst outcomes. Its mass is $\lambda(E) = 0.36$. The values of $x \in X$ where $u_2(x) \le 0.36$ are  $x \in [0, 0.1] \cup [0.9, 1]$. We evaluate Agent 1's utility on these  intervals to find the projection $\UPS(E)$:
\begin{enumerate}
    \item For $x \in [0, 0.1]$, $u_1(x) \in [0.99, 1]$.
    \item For $x \in [0.9, 1]$, $u_1(x) \in [0, 0.19]$.
\end{enumerate}
So,  $\UPS(E) = [0, 0.19] \cup [0.99, 1]$. Note that :
\[ \lambda(\UPS(E)) = 0.19 + (1 - 0.99) = 0.20 < \la(E)=0.36. \]
The condition $\lambda(E)\geq \lambda(\UPS(E))$ required by the theorem fails. The utility representations $u_1$ and $u_2$ are not a common cardinalization, for any measure on $X$.

Calculating canonical utility representations is quite easy. For Agent 1, $L_1(x) = \{z \in [0,1] : 1 - z^2 \le 1 - x^2\} = \{z \in [0,1] : z \ge x\} = [x, 1]$, we obtain that :
\[ \hat{u}_1(x) = \lambda([x, 1]) = 1 - x. \]
For Agent 2, $L_2(x) = \{z \in [0,1] : 1 - 4(z - 0.5)^2 \le 1 - 4(x - 0.5)^2\} = \{z \in [0,1] : |z - 0.5| \ge |x - 0.5|\}$.
This set excludes an open interval centered at $0.5$ with a radius of $d = |x - 0.5|$. Since the total domain is $[0,1]$, the excluded interval has a length of exactly $2d = 2|x - 0.5|$. The canonical utility is then:
\[ \hat{u}_2(x) = 1 - 2|x - 0.5| = 1 - |2x - 1|. \]

We substitute $x = 1 - \hat{u}_1$ into $\hat{u}_2(x)$ to parameterize the new utility possibility set $\hat{\UPS}$:
\[ \hat{u}_2 = 1 - |2(1 - \hat{u}_1) - 1| = 1 - |1 - 2\hat{u}_1| = 1 - |2\hat{u}_1 - 1|. \]
The set $\hat{\UPS}$ is an inverted V-shaped curve connecting $(0,0)$ to $(0.5, 1)$ to $(1,0)$. See the illustration on the right in Figure~\ref{fig:onedimex}.

We may now re-test the condition with $E = [0, 0.36]$.  Agent 2 achieves $\hat{u}_2(x) \le 0.36$ when:
\[ 1 - |2\hat{u}_1 - 1| \le 0.36 \implies |2\hat{u}_1 - 1| \ge 0.64 \]
This inequality holds when $2\hat{u}_1 - 1 \ge 0.64 \implies \hat{u}_1 \ge 0.82$, or $2\hat{u}_1 - 1 \le -0.64 \implies \hat{u}_1 \le 0.18$.  Now $\hat{\UPS}(E) = [0, 0.18] \cup [0.82, 1]$ with measure  $0.18 + (1 - 0.82) = 0.36$. So the condition $\lambda(E)\geq \lambda(\UPS(E))$ is satisfied.

The compromise alternative $x^*$ corresponds to the maximin of the canonical utilities $\hat{u_1}$ and $\hat{u_2}$. It turns out that $x^*$ is unique and equals $1/3$. Both players achieve a utility of $2/3$ at $x^*$ as illustrated on the right in Figure~\ref{fig:onedimex}.

In this example, the Markov kernel that we discussed in Section~\ref{sec:properties} is easy to calculate. Agent 2's utility $v_2$ is a deterministic function of $v_1$: $v_2 = f(v_1) = 2v_1$ on $[0,1/2]$ and $v_2 = f(v_1) = 2- 2v_1$ on $[1/2,1]$. Agent 1's utility is a random function of $v_2$: $v_1$ takes the values $v_2/2$ and $1-v_2/2$ with equal probability. Observe that these Markov kernels preserve the measure of any set of utils for one agent when mapped into a set of utils for the other agent. In this sense, they generalize the notion of the slope $=-1$ that we discussed in the introduction.

\subsection{A two-dimensional example}

Let $X = [0,1]^2$ and utilities be:
\begin{align*}
    u_1(x,y) = 1 - x \: \text{ and }    u_2(x,y) = y x^2.
\end{align*}

Both utilities map $X$ onto $[0,1]$. Because $u_1 = 1 - x$, we have $x = 1 - u_1$, which we may substitute into Agent 2's utility,
$u_2 = y (1 - u_1)^2$, and obtain that:
\[ \UPS = \{ (u_1, u_2) \in [0,1]^2 : 0 \le u_2 \le (1 - u_1)^2 \}. \]
The upper right-most boundary connecting $(0,1)$ to $(1,0)$ represents the Pareto frontier. 

Consider $E = [0.25, 1]$ with $\lambda(E) = 0.75$, the top $75\%$ of utility levels for Agent 2. In the Pareto frontier $u_2 = (1 - u_1)^2$; Agent 2 achieves $u_2 \ge 0.25$ when:
\[ (1 - u_1)^2 \ge 0.25 \implies 1 - u_1 \ge 0.5 \implies u_1 \le 0.5. \] So $\UPS(E) = [0, 0.5]$ with $\la(\UPS(E))= 1/2<0.75=\la(E)$. The condition fails.

We calculate the canonical utilities when $X$ is endowed with the Lebesgue measure. For Agent 1, the lower contour set at $x$ is the rectangle $[x, 1] \times [0, 1]$. So the canonical utility is :
\[ \hat{u}_1(x,y) = (1 - x) \cdot 1 = 1 - x .\]

For Agent 2,  let $c=u_2(x,y)=yx^2$ denote the utility level of 2 at $(x,y)$. The lower contour set $\{(x',y'):(y')(x')^2\leq c\}$ is bounded by $y' = 1$ for $x' \in [0, \sqrt{c}]$ and by the curve $y' = c / (x')^2$ for $x' \in [\sqrt{c}, 1]$. Integrating the area:
\begin{align*}
    \hat{u}_2(x,y) &= \int_0^{\sqrt{c}} 1 \, dx' + \int_{\sqrt{c}}^1 \frac{c}{(x')^2} \, dx' \\
    &= \sqrt{c} - c + \sqrt{c} = 2\sqrt{c} - c
\end{align*}
Substituting $c = yx^2$ yields the canonical utility: $\hat{u}_2(x,y) = 2x\sqrt{y} - yx^2$.

The Pareto frontier with the canonical utilities is easy to calculate:
\[ u_2 = 2(1 - u_1) - (1 - u_1)^2 = 2 - 2u_1 - (1 - 2u_1 + u_1^2) = 1 - u_1^2, \]
and the utility possibility set is bounded above by a concave parabola:
\[ \UPS = \{ (u_1, u_2) \in [0,1]^2 : 0 \leq u_2 \leq 1 - u_1^2 \} \]

We may check the condition in the theorem for the same set as before, $E = [0.25, 1]$. Agent 2 achieves $u_2 \ge 0.25$ when
\[ 1 - u_1^2 \ge 0.25 \implies u_1^2 \le 0.75 \implies u_1 \le \frac{\sqrt{3}}{{2}} \approx 0.866. \]
Then  $\UPS(E) = [0, 0.866]$, with a measure that exceeds $3/4=\la(E)$.

In the figure, we illustrate the original Pareto frontier as a red gray curve that lies below the diagonal. The new frontier (in solid black) is a concave curve. In both cases, the $\UPS$ is the locus of points below the Pareto frontier.

\begin{figure}[t]
    \centering
    \begin{tikzpicture}[scale=.8]
        \begin{axis}[
            width=0.7\textwidth,
            height=0.7\textwidth,
            xlabel={$u_1$},
            ylabel={$u_2$},
            xmin=0, xmax=1.1,
            ymin=0, ymax=1.1,
            axis lines=left,
                       grid style={line width=.1pt, draw=gray!30},
            major grid style={line width=.2pt,draw=gray!50},
            xtick={0, 0.25, 0.5, 0.75, 1},
            ytick={0, 0.25, 0.5, 0.75, 1},
            enlargelimits=false
        ]
        
                \fill[red, opacity=0.1] (0, 0.25) rectangle (1, 1);
        
                \addplot [domain=0:1, samples=100, fill=green!50!black, opacity=0.1, draw=none] {1 - x^2} \closedcycle;
        
                \addplot [domain=0:1, samples=100, very thick, red, dashed] {(1-x)^2};
        
                \addplot [domain=0:1, samples=100, very thick, black] {1 - x^2};
        
                \addplot [domain=0:1, samples=2, thick, black, dotted] {1 - x};
        
                \draw [very thick, black] (axis cs: 0, 0) -- (axis cs: 0, 1);
        
                \draw [very thick, black] (axis cs: 0, 0) -- (axis cs: 1, 0);

                \node[red, anchor=west,font=\footnotesize] at (axis cs: 0.7, 0.6) {$E=[0.25,1]$};
        \node[red, anchor=south west,font=\footnotesize] at (axis cs: 0.1, 0.07) {$u_2=(1-u_1)^2$};
               \node[red, anchor=south west,font=\footnotesize] at (axis cs: 0.1, 0.15) {Pareto frontier with $u$};
        \node[black, anchor=south west,font=\footnotesize] at (axis cs: 0.5, 0.8) {Pareto Frontier with $\hat{u}$};
        \node[black, anchor=south west,font=\footnotesize] at (axis cs: 0.5, 0.72) {$u_2=1-u_1^2$};

        \end{axis}
    \end{tikzpicture}
    \caption{The utility possibility space for the original and canonical utilities. The Pareto frontier at $u$ is $u_2=(1-u_1)^2$.}
\label{fig:2dex}
\end{figure}

In the example, one may wonder if the area below the ``anti-diagonal,'' the dotted line that connects $(0,1)$ and $(1,0)$, is also the $\UPS$ for some canonical utility. The answer is yes, but not for full support, absolutely continuous measure. Suppose that the measure on $X$ is $\nu$. Note that, because the marginal on each dimension must be uniform, the expectation of each agent's utility under the push-forward measure obtained from $\nu$ must be uniform. So $\E_{x\in \nu} [\tilde v_i] = 1/2$ for each agent $i$. If the $\UPS$ is the area below the anti-diagonal, then with probability $1$ we must have $\tilde v_1+\tilde v_2 = 1$. This means that $\nu$ has to be supported on the set of Pareto optimal outcomes, which here has Lebesgue measure 0. It also means that $\nu$ is not full support. 

In sum, the example illustrates how Theorem~\ref{thm:commonnu} delivers an ``atomless'' measure (since the marginals must be uniform), but not necessarily a measure under the hypotheses of Theorem~\ref{th:implementation}. In any case, the additional hypotheses are needed for the techniques used in the proof of Theorem~\ref{th:implementation}. They are not essential to understanding the economic content of assuming a common cardinalization, which is why we are interested in Theorem~\ref{thm:commonnu}. 

\end{document}

%% file: macros.tex
\makeatletter
\def\paragraph{\@startsection{paragraph}{4}%
  \z@\z@{-\fontdimen2\font}%
  {\normalfont\bfseries}}
\makeatother

\newcommand*\diff{\mathop{}\!\mathrm{d}}

\def\MU{\mathrm{MU}}

\renewcommand{\subset}{\subseteq}

\newlength{\twinwd}

\def\one{\mathbf{1}}

\def\LCS{L}
\def\UCS{U}

\def\Re{\mathbf{R}} \def\R{\mathbf{R}}

\def\E{\mathbf{E}}

\def\ep{\varepsilon}

\def\la{\lambda}
\def\da{\delta}
 
\def\phi{\varphi}
\def\sa{\sigma}

\def\one{\mathbf{1}}

\def\os{\emptyset}

\newcommand{\df}[1]{\textit{#1}}

\newcommand\abss[1]{\|#1 \|}

\DeclarePairedDelimiter{\norm}{\lVert}{\rVert}

\def\smallskip{\vspace{.3cm}}
\def\medskip{\vspace{.8cm}}
\def\bigskip{\vspace{1.5cm}}

\newtheorem{theorem}{Theorem}
\newtheorem{proposition}{Proposition}

\newtheorem{lemma}{Lemma}

\theoremstyle{remark}
\newtheorem{remark}{Remark}

\newtheorem{example}{Example}

\theoremstyle{definition}
\newtheorem{definition}{Definition}

%% file: main_arxiv.bbl
\begin{thebibliography}{51}
\newcommand{\enquote}[1]{``#1''}
\expandafter\ifx\csname natexlab\endcsname\relax\def\natexlab#1{#1}\fi

\bibitem[\protect\citeauthoryear{Abdulkadiro{\u{g}}lu, Pathak, and
  Roth}{Abdulkadiro{\u{g}}lu et~al.}{2005}]{abdulkadirouglu2005new}
\textsc{Abdulkadiro{\u{g}}lu, A., P.~A. Pathak, and A.~E. Roth} (2005):
  \enquote{The New York city high school match,} \emph{American Economic
  Review}, 95, 364--367.

\bibitem[\protect\citeauthoryear{Abreu and Sen}{Abreu and
  Sen}{1990}]{abreu1990subgame}
\textsc{Abreu, D. and A.~Sen} (1990): \enquote{Subgame perfect implementation:
  a necessary and almost sufficient condition,} \emph{Journal of Economic
  Theory}, 50, 285--299.

\bibitem[\protect\citeauthoryear{Aghion, Fehr, Holden, and Wilkening}{Aghion
  et~al.}{2018}]{aghion2018role}
\textsc{Aghion, P., E.~Fehr, R.~Holden, and T.~Wilkening} (2018): \enquote{The
  role of bounded rationality and imperfect information in subgame perfect
  implementation—an empirical investigation,} \emph{Journal of the European
  Economic Association}, 16, 232--274.

\bibitem[\protect\citeauthoryear{Aghion, Fudenberg, Holden, Kunimoto, and
  Tercieux}{Aghion et~al.}{2012}]{aghion2012subgame}
\textsc{Aghion, P., D.~Fudenberg, R.~Holden, T.~Kunimoto, and O.~Tercieux}
  (2012): \enquote{Subgame-perfect implementation under information
  perturbations,} \emph{Quarterly Journal of Economics}, 127, 1843--1881.

\bibitem[\protect\citeauthoryear{Agranov and Ortoleva}{Agranov and
  Ortoleva}{2017}]{agranov2017stochastic}
\textsc{Agranov, M. and P.~Ortoleva} (2017): \enquote{Stochastic choice and
  preferences for randomization,} \emph{Journal of Political Economy}, 125,
  40--68.

\bibitem[\protect\citeauthoryear{Alesina and Rosenthal}{Alesina and
  Rosenthal}{1995}]{AlesinaRosenthal1995}
\textsc{Alesina, A. and H.~Rosenthal} (1995): \emph{Partisan politics, divided
  government, and the economy}, Cambridge University Press.

\bibitem[\protect\citeauthoryear{Alesina and Tabellini}{Alesina and
  Tabellini}{1990}]{AlesinaTabellini1990}
\textsc{Alesina, A. and G.~Tabellini} (1990): \enquote{A positive theory of
  fiscal deficits and government debt,} \emph{Review of Economic Studies}, 57,
  403--414.

\bibitem[\protect\citeauthoryear{Ali, Kartik, and Kleiner}{Ali
  et~al.}{2023}]{alikartikkleiner23}
\textsc{Ali, S.~N., N.~Kartik, and A.~Kleiner} (2023): \enquote{Sequential Veto
  Bargaining With Incomplete Information,} \emph{Econometrica}, 91, 1527--1562.

\bibitem[\protect\citeauthoryear{Aliprantis and Border}{Aliprantis and
  Border}{2006}]{Aliprantis2006}
\textsc{Aliprantis, C.~D. and K.~C. Border} (2006): \emph{Infinite Dimensional
  Analysis: A Hitchhiker's Guide}, Berlin, Heidelberg: Springer, 3rd ed.

\bibitem[\protect\citeauthoryear{Anbarci}{Anbarci}{1993}]{anbarci1993noncooperative}
\textsc{Anbarci, N.} (1993): \enquote{Noncooperative foundations of the area
  monotonic solution,} \emph{Quarterly Journal of Economics}, 108, 245--258.

\bibitem[\protect\citeauthoryear{Anbarci and Bigelow}{Anbarci and
  Bigelow}{1994}]{anbarci1994area}
\textsc{Anbarci, N. and J.~P. Bigelow} (1994): \enquote{The area monotonic
  solution to the cooperative bargaining problem,} \emph{Mathematical Social
  Sciences}, 28, 133--142.

\bibitem[\protect\citeauthoryear{Barber{\`a} and Coelho}{Barber{\`a} and
  Coelho}{2022}]{barbera2022compromising}
\textsc{Barber{\`a}, S. and D.~Coelho} (2022): \enquote{Compromising on
  compromise rules,} \emph{RAND Journal of Economics}, 53, 95--112.

\bibitem[\protect\citeauthoryear{Border}{Border}{2015}]{BorderMax2015}
\textsc{Border, K.~C.} (2015): \enquote{Miscellaneous notes on optimization
  theory and related topics,} Mimeo, Caltech.

\bibitem[\protect\citeauthoryear{Border and Jordan}{Border and
  Jordan}{1983}]{border1983}
\textsc{Border, K.~C. and J.~S. Jordan} (1983): \enquote{Straightforward
  elections, unanimity and phantom Voters,} \emph{Review of Economic Studies},
  50, 153--170.

\bibitem[\protect\citeauthoryear{Border and Segal}{Border and
  Segal}{1994}]{border1994dynamic}
\textsc{Border, K.~C. and U.~Segal} (1994): \enquote{Dynamic consistency
  implies approximately expected utility preferences,} \emph{Journal of
  Economic Theory}, 63, 170--188.

\bibitem[\protect\citeauthoryear{Bouacida and Foucart}{Bouacida and
  Foucart}{2025}]{bouacida2025rituals}
\textsc{Bouacida, E. and R.~Foucart} (2025): \enquote{Rituals of reason:
  experimental evidence on the social acceptability of lotteries in allocation
  problems,} \emph{Games and Economic Behavior}, 152, 23--36.

\bibitem[\protect\citeauthoryear{Bowen, Chen, and Eraslan}{Bowen
  et~al.}{2014}]{bowen2014mandatory}
\textsc{Bowen, T.~R., Y.~Chen, and H.~Eraslan} (2014): \enquote{Mandatory
  versus discretionary spending: The status quo effect,} \emph{American
  Economic Review}, 104, 2941--2974.

\bibitem[\protect\citeauthoryear{Brams and Kilgour}{Brams and
  Kilgour}{2001}]{brams2001fallback}
\textsc{Brams, S.~J. and D.~M. Kilgour} (2001): \enquote{Fallback bargaining,}
  \emph{Group Decision and Negotiation}, 10, 287--316.

\bibitem[\protect\citeauthoryear{Budish and Cantillon}{Budish and
  Cantillon}{2012}]{budish2012multi}
\textsc{Budish, E. and E.~Cantillon} (2012): \enquote{The multi-unit assignment
  problem: Theory and evidence from course allocation at {H}arvard,}
  \emph{American Economic Review}, 102, 2237--2271.

\bibitem[\protect\citeauthoryear{Combe, Tercieux, and Terrier}{Combe
  et~al.}{2022}]{combe2022design}
\textsc{Combe, J., O.~Tercieux, and C.~Terrier} (2022): \enquote{The design of
  teacher assignment: Theory and evidence,} \emph{Review of Economic Studies},
  89, 3154--3222.

\bibitem[\protect\citeauthoryear{Dasgupta and Maskin}{Dasgupta and
  Maskin}{2007}]{dasgupta2007bargaining}
\textsc{Dasgupta, P. and E.~S. Maskin} (2007): \enquote{Bargaining and
  destructive power,} \emph{Annals of Economics \& Finance}, 8.

\bibitem[\protect\citeauthoryear{Davis, Hinich, and Ordeshook}{Davis
  et~al.}{1970}]{otto1970}
\textsc{Davis, O.~A., M.~J. Hinich, and P.~C. Ordeshook} (1970): \enquote{An
  expository development of a mathematical model of the electoral process,}
  \emph{American Political Science Review}, 64, 426--448.

\bibitem[\protect\citeauthoryear{Debreu}{Debreu}{1964}]{debreu1964continuity}
\textsc{Debreu, G.} (1964): \enquote{Continuity properties of Paretian
  utility,} \emph{International Economic Review}, 5, 285--293.

\bibitem[\protect\citeauthoryear{Dutta and Sen}{Dutta and
  Sen}{1991}]{dutta1991necessary}
\textsc{Dutta, B. and A.~Sen} (1991): \enquote{A necessary and sufficient
  condition for two-person Nash implementation,} \emph{Review of Economic
  Studies}, 58, 121--128.

\bibitem[\protect\citeauthoryear{Eraslan, Evdokimov, and Z{\'a}pal}{Eraslan
  et~al.}{2020}]{EraslanEvdokimovZapal2020}
\textsc{Eraslan, H., K.~Evdokimov, and J.~Z{\'a}pal} (2020): \enquote{Dynamic
  legislative bargaining,} ISER Discussion Paper 1090, Institute of Social and
  Economic Research, The University of Osaka.

\bibitem[\protect\citeauthoryear{Featherstone}{Featherstone}{2020}]{featherstone2020rank}
\textsc{Featherstone, C.~R.} (2020): \enquote{Rank efficiency: Modeling a
  common policymaker objective,} \emph{Unpublished paper, The Wharton School,
  University of Pennsylvania}.

\bibitem[\protect\citeauthoryear{Fehr and Schmidt}{Fehr and
  Schmidt}{1999}]{fehr1999theory}
\textsc{Fehr, E. and K.~M. Schmidt} (1999): \enquote{A theory of fairness,
  competition, and cooperation,} \emph{Quarterly journal of economics}, 114,
  817--868.

\bibitem[\protect\citeauthoryear{Folland}{Folland}{1999}]{Folland1999}
\textsc{Folland, G.~B.} (1999): \emph{Real Analysis: Modern Techniques and
  Their Applications}, New York: John Wiley \& Sons, 2nd ed.

\bibitem[\protect\citeauthoryear{Hurwicz and Schmeidler}{Hurwicz and
  Schmeidler}{1978}]{hurwicz1978construction}
\textsc{Hurwicz, L. and D.~Schmeidler} (1978): \enquote{Construction of outcome
  functions guaranteeing existence and {P}areto optimality of {N}ash
  equilibria,} \emph{Econometrica}, 1447--1474.

\bibitem[\protect\citeauthoryear{Hurwicz and Sertel}{Hurwicz and
  Sertel}{1999}]{hurwicz1999designing}
\textsc{Hurwicz, L. and M.~R. Sertel} (1999): \enquote{Designing mechanisms, in
  particular for electoral systems: the majoritarian compromise,} in
  \emph{Contemporary Economic Issues: Economic Behaviour and Design}, Springer,
  69--88.

\bibitem[\protect\citeauthoryear{Kalai and Smorodinsky}{Kalai and
  Smorodinsky}{1975}]{kalai1975other}
\textsc{Kalai, E. and M.~Smorodinsky} (1975): \enquote{Other solutions to
  Nash's bargaining problem,} \emph{Econometrica}, 513--518.

\bibitem[\protect\citeauthoryear{Klein and Thompson}{Klein and
  Thompson}{1984}]{kleinthompson}
\textsc{Klein, E. and A.~C. Thompson} (1984): \emph{Theory of correspondences},
  New York: Wiley.

\bibitem[\protect\citeauthoryear{Li}{Li}{2023}]{li2023bargaining}
\textsc{Li, X.} (2023): \enquote{Bargaining: Nash, consensus, or compromise?}
  \emph{Games and Economic Behavior}, 142, 730--742.

\bibitem[\protect\citeauthoryear{Maskin}{Maskin}{1999}]{maskin99}
\textsc{Maskin, E.} (1999): \enquote{Nash equilibrium and welfare optimality,}
  \emph{Review of Economic Studies}, 66, 23--38.

\bibitem[\protect\citeauthoryear{McCarty, Poole, and Rosenthal}{McCarty
  et~al.}{2006}]{McCartyPooleRosenthal2006}
\textsc{McCarty, N., K.~T. Poole, and H.~Rosenthal} (2006): \emph{Polarized
  America: The dance of ideology and unequal riches}, MIT press.

\bibitem[\protect\citeauthoryear{Moore and Repullo}{Moore and
  Repullo}{1988}]{moore1988subgame}
\textsc{Moore, J. and R.~Repullo} (1988): \enquote{Subgame perfect
  implementation,} \emph{Econometrica}, 1191--1220.

\bibitem[\protect\citeauthoryear{Moore and Repullo}{Moore and
  Repullo}{1990}]{moorerepulloNash}
---\hspace{-.1pt}---\hspace{-.1pt}--- (1990): \enquote{Nash implementation: A
  full characterization,} \emph{Econometrica}, 58, 1083--1099.

\bibitem[\protect\citeauthoryear{Moulin}{Moulin}{1984}]{moulin1984implementing}
\textsc{Moulin, H.} (1984): \enquote{Implementing the Kalai-Smorodinsky
  bargaining solution,} \emph{Journal of Economic Theory}, 33, 32--45.

\bibitem[\protect\citeauthoryear{Myerson}{Myerson}{1977}]{myerson1977two}
\textsc{Myerson, R.~B.} (1977): \enquote{Two-person bargaining problems and
  comparable utility,} \emph{Econometrica}, 1631--1637.

\bibitem[\protect\citeauthoryear{Nash et~al.}{Nash
  et~al.}{1950}]{nash1950bargaining}
\textsc{Nash, J.~F. et~al.} (1950): \enquote{The bargaining problem,}
  \emph{Econometrica}, 18, 155--162.

\bibitem[\protect\citeauthoryear{Nikzad}{Nikzad}{2022}]{nikzad2022rank}
\textsc{Nikzad, A.} (2022): \enquote{Rank-optimal assignments in uniform
  markets,} \emph{Theoretical Economics}, 17, 25--55.

\bibitem[\protect\citeauthoryear{Nunnari and Pozzi}{Nunnari and
  Pozzi}{2025}]{Nunnaripozzi2025}
\textsc{Nunnari, S. and M.~Pozzi} (2025): \enquote{Meta-analysis of
  distributional preferences,} \emph{Economic Journal, forthcoming}.

\bibitem[\protect\citeauthoryear{Owen and Daskin}{Owen and
  Daskin}{1998}]{OWEN1998423}
\textsc{Owen, S.~H. and M.~S. Daskin} (1998): \enquote{Strategic facility
  location: A review,} \emph{European Journal of Operational Research}, 111,
  423--447.

\bibitem[\protect\citeauthoryear{Rubinstein}{Rubinstein}{1982}]{rubinstein1982perfect}
\textsc{Rubinstein, A.} (1982): \enquote{Perfect equilibrium in a bargaining
  model,} \emph{Econometrica}, 97--109.

\bibitem[\protect\citeauthoryear{Rudin}{Rudin}{1991}]{rudin1991}
\textsc{Rudin, W.} (1991): \emph{Functional Analysis}, McGraw Hill Book Co.

\bibitem[\protect\citeauthoryear{Shapley}{Shapley}{1969}]{shapley1969utility}
\textsc{Shapley, L.~S.} (1969): \enquote{Utility comparison and the theory of
  games,} \emph{The Shapley Value. Essays in Honor of Lloyd S. Shapley},
  307--319.

\bibitem[\protect\citeauthoryear{Sj{\"o}str{\"o}m}{Sj{\"o}str{\"o}m}{1991}]{sjostrom1991necessary}
\textsc{Sj{\"o}str{\"o}m, T.} (1991): \enquote{On the necessary and sufficient
  conditions for Nash implementation,} \emph{Social Choice and Welfare}, 8,
  333--340.

\bibitem[\protect\citeauthoryear{Vartiainen}{Vartiainen}{2007{\natexlab{a}}}]{vartiainen2007nash}
\textsc{Vartiainen, H.} (2007{\natexlab{a}}): \enquote{Nash implementation and
  the bargaining problem,} \emph{Social Choice and Welfare}, 29, 333--351.

\bibitem[\protect\citeauthoryear{Vartiainen}{Vartiainen}{2007{\natexlab{b}}}]{vartiainen2007subgame}
---\hspace{-.1pt}---\hspace{-.1pt}--- (2007{\natexlab{b}}): \enquote{Subgame
  perfect implementation: a full characterization,} \emph{Journal of Economic
  Theory}, 133, 111--126.

\bibitem[\protect\citeauthoryear{Villani}{Villani}{2003}]{villani2003topics}
\textsc{Villani, C.} (2003): \emph{Topics in optimal transportation}, vol.~58,
  American Mathematical Society.

\bibitem[\protect\citeauthoryear{Woon and Anderson}{Woon and
  Anderson}{2012}]{woonanderson2012}
\textsc{Woon, J. and S.~Anderson} (2012): \enquote{Political bargaining and the
  timing of congressional appropriations,} \emph{Legislative Studies
  Quarterly}, 37, 409--436.

\end{thebibliography}
